\documentclass[nolineno]{biometrika}
 \usepackage{stmaryrd}

\usepackage{mathrsfs,amsfonts,amssymb}
\usepackage[dvipsnames]{xcolor}
\usepackage{tabularx,caption,tabulary,graphicx,multirow,booktabs,relsize, fancyhdr} 

\usepackage{adjustbox}
\usepackage{amsmath}
\usepackage{epsfig}
\usepackage{graphicx}
\usepackage{subfigure}
\usepackage{times}
\usepackage{bm}
\usepackage{natbib}
 \usepackage{color}

\def\T{{ \mathrm{\scriptscriptstyle T} }}

\def\pr{{\mathrm{pr}}}

\def\part{\partial}

\newcommand{\cv}{\mathrm{cv}}
\newcommand{\vecd}{\mathrm{vec}}
\newcommand{\diagn}{\mathrm{diag}}
\newcommand{\tr}{\mbox{tr}}

\def\JRSSB{{\sl Journal of the Royal Statistical Society}, {\bf B}}

\def\JASA{{\sl Journal of the American Statistical Association}}

\def\SS{{\sl Statistica Sinica}}

\def\AS{{\sl The Annals of Statistics}}

\def\ET{{\sl Econometric Theory}}
\def\JOE{{\sl Journal of Econometrics}}

\begin{document}

\jname{Biometrika}
\jyear{} \jvol{} \jnum{}
\accessdate{Advance Access publication on } \copyrightinfo{\Copyright\ 2010 Biometrika Trust\goodbreak {\em Printed in Great Britain}}


\markboth{J. Chang, Q. Yao \and W. Zhou}{Testing for vector white noise}

\title{Testing for high-dimensional white noise using maximum cross correlations}

\author{JINYUAN CHANG }
\affil{School of Statistics, Southwestern University of Finance and Economics, Chengdu,
Sichuan 611130, China \email{changjinyuan@swufe.edu.cn} }

\author{QIWEI YAO}
\affil{Department of Statistics, London School of Economics and Political Science, WC2A 2AE, UK \email{q.yao@lse.ac.uk}}

\author{\and WEN ZHOU}
\affil{Department of Statistics, Colorado State University, Fort Collins, Colorado 80523, USA \email{riczw@stat.colostate.edu}}

\maketitle

\begin{abstract}
We propose a new omnibus test for vector white noise using the maximum
absolute autocorrelations and cross-correlations of the component
series. Based on an approximation by the $L_\infty$-norm  of a normal
random vector, the critical value of the test can be evaluated by
bootstrapping from a multivariate normal distribution. In contrast to the
conventional white noise test, the new method is proved to be valid for
testing the departure from white noise that is not independent and
identically distributed. We illustrate the accuracy and the power of the
proposed test by simulation, which also shows that the new test
outperforms several commonly used methods including, for example, the
Lagrange multiplier test and the multivariate Box--Pierce portmanteau
tests, especially when the dimension of time series is high in relation
to the sample size.  The numerical results also indicate that the
performance of the new test can be further enhanced when it is applied
to pre-transformed data obtained via the time series  principal component
analysis proposed  by Chang, Guo and
Yao (arXiv:1410.2323). The proposed procedures
have been implemented in an \texttt{R} package.
\end{abstract}


\section{Introduction}

Testing for white noise or serial correlation is a fundamental problem in
statistical inference, as many testing problems in linear modelling can
be transformed into a white noise test.
Testing for white noise is often pursued in two different manners: (i) the
departure from white noise is specified as an alternative hypothesis in
the form of an explicit parametric family such as an autoregressive
moving average model, and (ii) the alternative hypothesis is unspecified.
With an explicitly specified alternative, a likelihood ratio test can be
applied. Likelihood-based tests typically have more power to detect a
specific form of the departure 
than omnibus
tests which try to detect arbitrary departure from white noise. The
likelihood approach has been taken further in the nonparametric context using
the generalized likelihood ratio test initiated by \cite{FZZ01}; see
Section 7.4.2 of \cite{FY03} and also \cite{FZ04}. Nevertheless many
applications including model diagnosis do not lead to a natural
alternative model. Therefore various omnibus tests, especially the
celebrated Box--Pierce test and its variants, remain popular. Those
portmanteau tests are proved to be asymptotically $\chi^2$-distributed
under the null hypothesis, which makes their application extremely easy.
 See Section 3.1 of \cite{Li_2004} and Section 4.4 of
\cite{Lutkepohl_2005} for further information on those portmanteau tests.

While portmanteau tests are designed for testing
 white noise, their asymptotic $\chi^2$-distributions are established
under the assumption that observations under the null hypothesis are
independent and identically distributed.
However, empirical evidence, including that in Section~\ref{se:numer}
below, suggests that this may represent another case in which the theory
is more restrictive than the method itself. Asymptotic theory of  portmanteau tests for white noise that is not independent and identically distributed has
attracted a lot of attention. One of the most popular approaches is to
establish the asymptotic normality of a normalized portmanteau test
statistic. An incomplete list in this endeavour includes
\cite{Durlauf91}, \cite{RT96}, \cite{Deo00}, \cite{Lo2001}, \cite{FRZ05}, \cite{EL09} and
\cite{Shao11}. However, the convergence is typically
slow.
\cite{Horowitz06} proposed a double blockwise
bootstrap method to test for white noise that is not independent and identically distributed.

In this paper we propose a new omnibus test for vector white noise.
Instead of using a portmanteau-type statistic, the new test is based on
the maximum absolute auto- and cross-correlation of all component time
series. This avoids the impact of small correlations.
When most
auto- and cross-correlations are small, the
Box--Pierce tests have too many
degrees of freedom in their asymptotic distributions.
In contrast the new test performs well when there is at least one large
absolute auto- or cross-correlation at a non-zero lag.
The null distribution of the maximum correlation test statistic can be
approximated asymptotically by that of $|G|_\infty$, where $G$ is a
Gaussian random vector, and $|u|_\infty = \max_{1\leq i\leq s} |u_i|$
denotes the $L_\infty$-norm of a vector $u = (u_1, \ldots, u_s)^\T$. Its critical
values can therefore be evaluated by bootstrapping from a
multivariate normal distribution.

An added advantage of the new test is its ability to handle high-dimensional
series, in the sense that the number of series is as large as, or
even larger than, their length. Nowadays, it is common to model and forecast many time series at once, which
has direct applications in, among others,  finance, economics, environmental and
medical studies. The current literature on high-dimensional
time series  focuses on estimation
and dimension-reduction aspects.
See, for example, \cite{BM15}, and \cite{GWY16}
and the references within for high-dimensional vector autoregressive models,
and \cite{BN02}, \cite{Fornietal_2005},
\cite{LY12} and \cite{ChangGuoYao_2015}
for high-dimensional time series factor models.
The model diagnostics
has largely been untouched, as far as we are aware. The test proposed in this paper represents an effort to fill in this gap.

We compare the performance of the new test with those of the three Box--Pierce
types of portmanteau tests, the Lagrange multiplier test and a likelihood
ratio test in simulation, which shows that the new test attains
the nominal significance levels more accurately and is also more powerful when the
dimension of time series is large or moderately large. Its performance can be further enhanced by first applying time series principal component analysis, proposed by
 Chang, Guo and Yao (arXiv:1410.2323).

Let $\otimes$ and $\vecd$ denote, respectively, the Kronecker product and the
vectorization for matrices, $I_{s}$ be the $s\times s$ identity
matrix, and $|A|_\infty=\max_{1\leq i\leq
\ell,1\leq j\leq m}|a_{ij}|$ for an $\ell \times m$ matrix $A\equiv (a_{i,j})$.
Denote by $\lceil x\rceil$ and $\lfloor x\rfloor$, respectively, the
smallest integer not less than $x$ and the largest integer not greater than $x$.

\section{Methodology}

\label{method}

\subsection{Tests}\label{se:test}
\label{sec21}
 Let $\{\varepsilon_t\}$ be a $p$-dimensional weakly stationary time
series with mean zero. Denote by
$  \Sigma(k)=\textrm{cov}(\varepsilon_{t+k},\varepsilon_t)$ and $
\Gamma(k) = \textrm{diag}\{{\Sigma}(0)\}^{-1/2}
{\Sigma}(k)\textrm{diag}\{{\Sigma}(0)\}^{-1/2},
$
respectively,
the autocovariance and the autocorrelation of $\varepsilon_t$ at lag $k$,
where $\textrm{diag}(\Sigma)$ denotes the diagonal matrix consisting of
the diagonal elements of $\Sigma$ only. When $\Sigma(k)\equiv 0$ for all $k\ne 0$,
$\{\varepsilon_t\}$ is white noise.

 With the available observations
$\varepsilon_1, \ldots, \varepsilon_n$, let
\begin{equation}\label{eq:autocorr}
\widehat{\Gamma}(k) \equiv \{\widehat{\rho}_{ij}(k)\}_{1\leq i,j\leq p} =\textrm{diag}\{\widehat{\Sigma}(0)\}^{-1/2} \widehat{\Sigma}(k)\textrm{diag}\{\widehat{\Sigma}(0)\}^{-1/2}
\end{equation}
be the sample autocorrelation matrix at lag $k$, where
\begin{equation}\label{eq:auto}
\widehat{\Sigma}(k)=\frac{1}{n}\sum_{t=1}^{n-k}\varepsilon_{t+k}\varepsilon_{t}^\T
\end{equation}
is the sample autocovariance matrix.

Consider the hypothesis testing problem
\begin{equation}\label{eq:test}
H_0:\{\varepsilon_t\}~\mbox{is white noise}~~~\textrm{versus}~~~H_1:\{\varepsilon_t\}~\mbox{is not white noise}.
\end{equation}
Since $\Gamma(k) \equiv 0$ for any $k\geq1$ under $H_0$, our test statistic $T_n$ is defined as
\begin{equation} \label{Tn}
T_n=\max_{1\leq k\leq K}T_{n,k},
\end{equation}
where $T_{n,k}=\max_{1\leq i, j\leq p}{n}^{1/2}|\widehat{\rho}_{ij}(k)|$ and $K\ge 1$ is a prescribed integer. We reject $H_0$ if $T_n>\textrm{cv}_\alpha$, where $\cv_\alpha >0$ is the critical value determined by
\begin{equation}\label{eq:cv}
\pr(T_n> \cv_\alpha) = \alpha
\end{equation}
under $H_0$, and $\alpha \in(0,1)$ is the significance level of the test.

To determine $\cv_\alpha$, we need to derive the distribution of $T_n$
under $H_0$.
 Proposition \ref{tm:1} below shows that the Kolmogorov
distance between this distribution and that of the $L_\infty$-norm of a
$N(0, \Xi_n)$ random vector
 converges to zero,
even when $p$ diverges at an exponential rate of $n$, where
\begin{equation}\label{eq:XI}
\Xi_n= (I_K\otimes W){E}(\xi_n\xi_n^\T)(I_K\otimes W),
\end{equation}
$$\xi_n={n}^{1/2}(\vecd\{\widehat{\Sigma}(1)\}^\T,\ldots,
\vecd\{\widehat{\Sigma}(K)\}^\T)^\T, \;\;\;
W=\diagn\{\Sigma(0)\}^{-1/2}\otimes\diagn\{\Sigma(0)\}^{-1/2}.$$
This paves the way to evaluate $\cv_\alpha$ simply by drawing a bootstrap
sample from $N(0, \widehat{\Xi}_n)$, where $\widehat{\Xi}_n$ is an
appropriate estimator for $\Xi_n$.

\begin{proposition}\label{tm:1}
Let Conditions {\rm\ref{as:var}--\ref{as:cov}} in Section
{\rm\ref{se:theory}} below hold and $G\sim N(0,\Xi_n)$. 
There exists a positive constant $\delta_1$ depending only on the
constants appeared in Conditions {\rm\ref{as:var}--\ref{as:cov}} for
which $\log p\leq Cn^{\delta_1}$ for some constant $C>0$. Then it holds under $H_0$ that
\[
\sup_{s\geq 0}\big|\pr(T_n>s)-\pr(|G|_\infty>s)\big|\rightarrow0, \qquad {\rm as} \;
n\rightarrow\infty.
\]
\end{proposition}

By replacing $\Xi_n$ in \eqref{eq:XI} by $\widehat{\Xi}_n$, where $\widehat{\Xi}_n$ is defined in Section \ref{se:est} below, the critical value $\cv_\alpha$ in (\ref{eq:cv}) can be replaced by $\widehat{\cv}_\alpha$ which is determined by
\begin{equation} \label{eq:hatcv}
\pr(|G|_\infty > \widehat{\cv}_\alpha) = \alpha, 
\end{equation}
where $G \sim N(0,\widehat{\Xi}_n)$. In practice, we can draw
$G_1,\ldots, G_{B}$ independently from $N(0,\widehat{\Xi}_n)$ for a large
integer $B$.
The $\lfloor B\alpha\rfloor$-th largest value among $|G_1|_\infty,\ldots, |G_{B}|_\infty$ is taken as the critical value $ \widehat{\cv}_\alpha$. We then reject $H_0$ whenever $T_n > \widehat{\cv}_\alpha$.

\begin{remark} \label{remark1}
When $p$ is large or moderately large, it is advantageous to apply the
time series principal component analysis
proposed  in 
arXiv:1410.2323 to the data first. 
We denote by $T_n^*$ the resulted test. More
precisely, we compute an invertible transformation matrix $Q$  using the {\tt R} function
{\tt segmentTS}  in the package {\tt PCA4TS} available at {\tt CRAN}.
Then $T_n^*$ is defined in the same manner as $T_n$ in (\ref{Tn}) with
$\{\varepsilon_1, \ldots, \varepsilon_n\}$ replaced by
$\{\varepsilon_1^*, \ldots, \varepsilon_n^*\}$, where $\varepsilon_t^*
= Q \varepsilon_t$.
As $Q$ does not depend on $t$,
$\{\varepsilon_t,\, t\ge 1\}$ is white noise if and only
if $\{\varepsilon_t^*, \, t\ge 1 \}$
is white noise. The time series principal
component analysis makes the component autocorrelations
as large as possible by suppressing the
cross-correlations among different components at all time lags.
This makes the maximum correlation
greater, and therefore the test more powerful. See also
the simulation results in Section~\ref{se:numer}.
\end{remark}

\subsection{Estimation of $\Xi_n$}\label{se:est}

By Lemma 3.1 of \cite{CCK_2013}, the proposed test in Section \ref{se:test} is valid if the estimator $\widehat{\Xi}_n$ satisfies $|\widehat{\Xi}_n-\Xi_n|_\infty=o_p(1)$. We construct such an estimator now even when the dimension of time series is ultra-high, i.e. $p\gg n$. Let $\tilde{n}=n-K$ and
\begin{equation} \label{fft}
f_t=\{\vecd(\varepsilon_{t+1}\varepsilon_t^\T),\ldots,\vecd(\varepsilon_{t+K}\varepsilon_t^\T)\}^\T~~~(t=1,\ldots,\tilde{n}).
\end{equation}
The second factor $E(\xi_n\xi_n^\T)$ on the right-hand side of (\ref{eq:XI}) is closely related to $\textrm{var}(\tilde{n}^{-1/2}\sum_{t=1}^{\tilde{n}}f_t)$, the long-run covariance of $\{f_t\}_{t=1}^{\tilde{n}}$. The long-run covariance plays an important role in the inference with dependent data. There exist  various estimation methods for long-run covariances, including the kernel-type estimators \citep{A91}, and the estimators utilizing the moving block bootstraps \citep{L03}. See also   \cite{DL97} and \cite{KVB00}. We adopt the kernel-type estimator for the long-run covariance of $\{f_t\}_{t=1}^{\tilde{n}}$ 
 \begin{equation}\label{ae}
 \widehat{J}_n = \sum_{j=-\tilde{n}+1}^{\tilde{n}-1} \mathcal{K}\left(\frac{j}{b_n}\right) \widehat{H}(j),
 \end{equation}
 where $\widehat{H}(j)=\tilde{n}^{-1}\sum_{t=j+1}^{\tilde{n}} f_t f_{t-j}^\T$ if $j\geq 0$ and $\widehat{H}(j)=\tilde{n}^{-1}\sum_{t=-j+1}^{\tilde{n}} f_{t+j}f_{t}^\T$ otherwise, $\mathcal{K}(\cdot)$ is a symmetric kernel function that is continuous at $0$ with $\mathcal{K}(0)=1$, and $b_n$ is the bandwidth diverging with $n$. Among a variety of kernel functions that guarantee the positive definiteness of the long-run covariance estimators, \cite{A91} derived an optimal kernel, i.e. the quadratic spectral kernel
\begin{equation}\label{qs}
\mathcal{K}_{QS}(x)=\frac{25}{12\pi^2x^2}\left\{\frac{\sin (6\pi x/5)}{6\pi x/5}-\cos(6\pi x/5)\right\}
\end{equation}
by minimizing the asymptotic truncated mean square error of the
estimator. For the numerical study in Section \ref{se:numer}, we always
use this kernel function with an explicitly specified bandwidth selection
procedure. The theoretical results in Section \ref{se:theory} apply to
general kernel functions. As now $\widehat{J}_n$ in (\ref{ae}) provides
an estimator for $E(\xi_n\xi_n^\T)$, $\Xi_n$ in (\ref{eq:XI}) can be estimated by
\begin{equation*} \label{eq:hatXI}
\widehat{\Xi}_n=(I_K\otimes \widehat{W})\widehat{J}_n(I_K\otimes \widehat{W}),
\end{equation*}
where $\widehat{W}=\diagn\{\widehat{\Sigma}(0)\}^{-1/2} \otimes\diagn\{\widehat{\Sigma}(0)\}^{-1/2}$ for $\widehat{\Sigma}(0)$ defined in (\ref{eq:auto}). Simulation results show that the proposed test with this estimator performs very well.

\subsection{Computational issues}\label{se:genera}

To draw a random vector $G\sim N(0,\widehat{\Xi}_n)$, the standard
approach consists of three steps: (i) perform the Cholesky decomposition
for the $p^2K\times p^2K$ matrix $\widehat{\Xi}_n =L^\T L$, (ii) generate
$p^2K$ independent $N(0,1)$ random variables $z=(z_1, \ldots,
z_{p^2K})^\T$,  (iii) perform transformation $G = L^\T z$. Computationally this
is an $(np^4K^2+p^6K^3)$-hard problem requiring a large storage space
for $\{f_t\}_{t=1}^{\tilde
n}$ and  matrix $\widehat{\Xi}_n$.
To circumvent the high computing cost with large $p$ and/or $K$,
we propose a method below which requires to generate random variables from an
$\tilde{n}$-variate normal distribution instead.

Let $\Theta$ be an $\tilde{n}\times \tilde{n}$ matrix with the $(i,j)$-th
element $\mathcal{K}\{(i-j)/b_n\}$. Let
$\eta=(\eta_1,\ldots,\eta_{\tilde{n}})^\T\sim N(0, \Theta)$ be a random
vector independent of
$\{\varepsilon_1,\ldots,\varepsilon_n\}$. Then it is easy
to see that conditionally on $\{\varepsilon_1,\ldots,\varepsilon_n\}$,
\begin{equation} \label{nn}
G=(I_K\otimes \widehat{W})\left(\frac{1}{\surd\tilde{n}}\sum_{t=1}^{\tilde{n}} \eta_t f_t\right)\sim N(0,\widehat{\Xi}_n).
\end{equation}
Thus a random sample from $N(0,\widehat{\Xi}_n)$ can be obtained from a random sample
from $N(0, \Theta)$ via \eqref{nn}. The computational complexity of the new method is only $O(n^3)$, independent of $p$ and $K$. The required storage space is also much smaller.

\section{Theoretical properties}\label{se:theory}

Write $\varepsilon_t=(\varepsilon_{1,t},\ldots,\varepsilon_{p,t})^\T$ for each $t=1,\ldots,n$. To investigate the theoretical properties of the proposed testing procedure, we need the following regularity conditions.
\begin{condition}\label{as:var}
There exists a constant $C_1>0$ independent of $p$ such that $\textrm{var}(\varepsilon_{i,t})\geq C_1$ uniformly holds for any $i=1,\ldots,p$.
\end{condition}

\begin{condition}\label{as:tail}
There exist three constants $C_2, C_3>0$ and $r_1\in(0,2]$ independent of $p$ such that
$
\sup_t\sup_{1\leq i\leq p}\pr(|\varepsilon_{i,t}|>x)\leq C_2\exp(-C_3x^{r_1})
$
for any $x>0$.
\end{condition}

\begin{condition}\label{as:betam}
Assume that $\{\varepsilon_t\}$ is $\beta$-mixing in the sense that
$
\beta_k \equiv\sup_t E \{\sup_{B\in\mathcal{F}_{t+k}^\infty} \big|\pr(B\mid\mathcal{F}_{-\infty}^t)-\pr(B)\big|\} \to 0
$
as $k \to \infty$, where $\mathcal{F}_{-\infty}^u$ and $\mathcal{F}_{u+k}^\infty$ are the $\sigma$-fields generated respectively by $\{\varepsilon_t\}_{t\leq u}$ and $\{\varepsilon_t\}_{t\geq u+k}$. Furthermore there exist two constants $C_4>0$ and $r_2\in(0,1]$ independent of $p$ such that $ \beta_k\leq \exp(-C_4k^{r_2}) $
 for all $k\geq1$.
\end{condition}

\begin{condition}\label{as:cov}
There exists a constant $C_5>0$ and $\iota>0$ independent of $p$ such that
\[
\begin{split}
C_5^{-1}<&~\liminf_{q\rightarrow\infty}\inf_{m\geq 0} E\bigg(\bigg|\frac{1}{q^{1/2}}\sum_{t=m+1}^{m+q}\varepsilon_{i,t+k}\varepsilon_{j,t}\bigg|^{2+\iota}\bigg)\\
\leq&~ \limsup_{q\rightarrow\infty}\sup_{m\geq 0}E\bigg(\bigg|\frac{1}{q^{1/2}}\sum_{t=m+1}^{m+q}\varepsilon_{i,t+k}\varepsilon_{j,t}
\bigg|^{2+\iota}\bigg)<C_5,~~(i,j=1,\ldots,p; k=1,\ldots,K).
\end{split}
\]
\end{condition}

Condition \ref{as:var} ensures that all component series are not
degenerate. Condition \ref{as:tail} is a common assumption in the
literature on ultra high-dimensional data analysis. It ensures
exponential-type upper bounds for the tail probabilities of the
statistics concerned. The $\beta$-mixing assumption in Condition
\ref{as:betam} is mild. Causal autoregressive moving average processes with continuous innovation
distributions are $\beta$-mixing with exponentially decaying $\beta_k$.
So are the stationary Markov chains satisfying certain conditions. See
Section 2.6.1 of \cite{FY03} and the references within. In fact
stationary generalized autoregressive conditional heteroskedasticity models with finite second moments and continuous
innovation distributions are also $\beta$-mixing with exponentially
decaying $\beta_k$; see Proposition 12 of \cite{CC02}. If we only require
$\sup_t\sup_{1\leq i\leq
p}\pr(|\varepsilon_{i,t}|>x)=O\{x^{-2(\nu+\epsilon)}\}$ for any $x>0$ in
Condition \ref{as:tail} and $\beta_k=O\{k^{-\nu(\nu+\epsilon)/(2\epsilon)}\}$ in
Condition \ref{as:betam} for some $\nu>2$ and $\epsilon>0$, we can apply
Fuk--Nagaev type inequalities to construct the upper bounds for the
tail probabilities of the statistics for which our testing procedure
still works for $p$ diverging at some polynomial rate of $n$. We refer to
Section 3.2 of
arXiv:1410.2323 for the implementation of
Fuk--Nagaev type inequalities in such a scenario. The
$\beta$-mixing condition 
can be replaced by the
$\alpha$-mixing condition under which 
we can justify the proposed method for $p$
diverging at some polynomial rate of $n$ by using Fuk--Nagaev type
inequalities. However, it remains open to establish the relevant
properties under $\alpha$-mixing for $p$ diverging
at some exponential rate of $n$.
Condition \ref{as:cov} is a
technical assumption for the validity of the Gaussian approximation for dependent data.

Our main asymptotic results indicate that the critical value  $\widehat{\cv}_\alpha$ defined in (\ref{eq:hatcv}) by the normal approximation is asymptotically valid, and, furthermore, the proposed test is consistent.

\begin{theorem}\label{tm:2}
Let Conditions {\rm\ref{as:var}--\ref{as:cov}} hold,
$|\mathcal{K}(x)|\asymp |x|^{-\tau}$ as
$|x|\rightarrow\infty$ for some $\tau>1$, and $b_n\asymp
n^{\rho}$ for some $0<\rho<\min\{(\tau-1)/(3\tau),r_2/(2r_2+1)\}$.
Let $\log
p\leq Cn^{\delta_2}$ for some positive constants $ \delta$ and $C$ depending
on the constants in Conditions {\rm\ref{as:var}--\ref{as:cov}} only.
Then it holds  under $H_0$
that 
\[
\pr(T_n>\widehat{\cv}_\alpha)\rightarrow\alpha, \quad  \; n \to \infty.
\]
\end{theorem}

\begin{theorem}\label{tm:3}
Assume that the conditions of Theorem {\rm\ref{tm:1}} hold. Let $\varrho$ be the largest element in the main diagonal of $\Xi_n$, and $\lambda(p,\alpha)=\{2\log(p^2K)\}^{1/2}+\{2\log(1/\alpha)\}^{1/2}$. Suppose that
\[
\max_{1\leq k\leq K}\max_{1\leq i,j\leq p}|\rho_{i,j}(k)|\geq \varrho^{1/2}(1+\epsilon_n)n^{-1/2}\lambda(p,\alpha)
\]
for some positive $\epsilon_n$ satisfying $\epsilon_n\rightarrow0$ and
$\epsilon_n^2\log p\rightarrow\infty$. Then it holds  under $H_1$ that
\[
\pr(T_n> \widehat{\cv}_\alpha)\rightarrow1,  \quad  \; n \to \infty.
\]
\end{theorem}

\section{Numerical properties}
\label{se:numer}

\subsection{Preliminary}
\label{se:41}
In this section, we illustrate the finite sample properties of the
proposed test $T_n$ by simulation. Also included is the test $T_n^*$
based on the pre-transformed data as stated in Remark~\ref{remark1} in
Section~\ref{se:test}. We always use the quadratic spectral kernel
$\mathcal{K}_{QS}(x)$ specified in (\ref{qs}). In addition, we always use
the data-driven bandwidth $b_n=1.3221\{\widehat{a}(2)\tilde{n}\}^{1/5}$
suggested in Section 6 of \cite{A91}, where
$
\widehat{a}(2)=\{\sum_{\ell=1}^{p^2K} {4\widehat{\rho}_\ell^2\widehat{\sigma}_\ell^4} (1-\widehat{\rho}_\ell)^{-8}\}\{\sum_{\ell=1}^{p^2K} {\widehat{\sigma}_\ell^4}{ (1-\widehat{\rho}_\ell)^{-4}}\}^{-1}
$
with $\widehat{\rho}_\ell$ and $\widehat{\sigma}_\ell^2$ being,
respectively, the estimated autoregressive coefficient and innovation
variance from fitting an AR(1) model to time series
$\{f_{\ell,t}\}_{t=1}^{\tilde{n}}$, where $f_{\ell,t}$ is the $\ell$-th
component of $f_t$ defined in (\ref{fft}). 
We draw $G_1,\ldots,G_B$ independently from $N(0,\widehat{\Xi}_n)$ with $B=2000$
based on (\ref{nn}) and take the $\lfloor B\alpha\rfloor$-th largest
value among $|G_1|_\infty,\ldots,|G_B|_\infty$ as the critical value
$\widehat{\cv}_\alpha$. We set the nominal significance level at
$\alpha=0.05$, $n=300$, $p=3,15,50,150$, and $K=2,4,6,8,10$. For
each setting, we replicate the experiment $500$ times.

We compare the new tests $T_n$ and $T_n^*$ with three multivariate
portmanteau tests with test statistics:
$Q_{1}  = n \sum_{k=1}^K \tr\{\widehat{\Gamma}(k)^\T
\widehat{\Gamma}(k)\}$ \citep{BP70},  $Q_{2}  =
n^2 \sum_{k=1}^K\tr\{\widehat{\Gamma}(k)^\T \widehat{\Gamma}(k)\}/(n-k)$
\citep{Hosking_1980}, and $Q_{3}  = n \sum_{k=1}^K
\tr\{\widehat{\Gamma}(k)^\T\widehat{\Gamma}(k)\} + p^2K(K+1)/(2n)$
\citep{LM_1981}, where $\widehat{\Gamma}(k)$ is the sample correlation
matrix (\ref{eq:autocorr}).
Also, we compare $T_n$ and $T_n^*$ with the Lagrange
multiplier test \citep{Lutkepohl_2005}, as well
as a likelihood ratio test proposed by \cite{TB81}. The test of \cite{TB81} is designed for testing for a vector autoregressive model of order $r$ against that of order $r+1$
and is therefore applicable for testing \eqref{eq:test}
with $r=0$. In particular, different from all the other tests included in
the comparison, the test of \cite{TB81} does not involve
the lag parameter $K$.
For those tests relying on the asymptotic
$\chi^2$-approximation, it is known that the $\chi^2$-approximation
is poor when the degree of freedom is large. In our simulation, we
perform those tests based on the normal approximation instead when $p>10$.
Further discussions on those tests are referred to Section 3.1 of
\cite{Li_2004} and Section 4.4 of \cite{Lutkepohl_2005}. The new tests
$T_n$ and $T_n^*$, together with the aforementioned other tests, have been
implemented in an \texttt{R} package \texttt{HDtest} currently
available online at \texttt{CRAN}.

\subsection{Empirical sizes}

To examine the approximations for significance levels of the tests,
we generate data from the white noise model $\varepsilon_t=A z_t$,
where $\{z_t\}$ is a $p\times 1$ white noise. We consider
three different loading matrices for $A$ as following.
\begin{description}
\item
Model 1: Let $S=(s_{k\ell})_{1\leq k,\ell\leq p}$ for $s_{k \ell}=0.995^{|k-\ell|}$, then let $A=S^{1/2}$.

\item
Model 2: Let $r=\lceil p/2.5\rceil$, $S=(s_{k\ell})_{1\leq k,\ell\leq p}$ where $s_{kk}=1$, $s_{k\ell}=0.8$ for $r(q-1)+1\leq k\neq \ell\leq rq$ for $q=1,\ldots,   \lfloor p/r \rfloor  $, and $s_{k \ell}=0$ otherwise. Let $A=S^{1/2}$ which is a block diagonal matrix.

\item
Model 3: Let $A=(a_{k\ell})_{1\leq k,\ell\leq p}$ with $a_{k\ell}$'s being independently generated from  $U(-1,1)$.
\end{description}
We consider the two types of white noise: (i) $z_t,\,  t\ge 1,$ are
independent and $N(0, I_p)$, and (ii) $z_t$
consists of $p$ independent
autoregressive conditionally heteroscedastic processes, i.e. each component process
is of the form $u_{t} = \sigma_{t} e_{t}$, where $e_{t}$ are independent and $N(0, 1)$,
and $\sigma_{t}^2 = \gamma_0 + \gamma_1 u_{t-1}^2$ with $\gamma_0$ and $\gamma_1$
generated from, respectively, $U(0.25, 0.5)$ and $U(0, 0.5)$ independently for
different component processes. Experiments with more complex white noise processes
are reported in the Supplementary Material.

\begin{table}[h!]
   {
   \centering
\caption{The empirical sizes ($\%$)  of the tests $T_n$, $T_n^*$,
$Q_{1},  Q_{2}$, $ Q_{3}$, Lagrange multiplier test (LM) and Tiao \& Box'
test (TB) for testing white noise $\varepsilon_t = A z_t$ at the $5\%$ nominal
level, where $ z_t, \, t\ge 1,$ are independent and $N(0, I_p)$.}
      \noindent\makebox[1\textwidth]{%
      \begin{adjustbox}{max width=1\textwidth}
\begin{tabular}{cc |ccccccc |ccccccc| ccccccccc}
& &\multicolumn{7}{ c| }{Model 1}& \multicolumn{7}{ c| }{Model 2} & \multicolumn{7}{ c }{Model 3}   \\
 $p$& $K$&  $T_n$ &$T_n^*$& $Q_{1}$ & $Q_{2}$ & $Q_{3}$ & LM & TB & $T_n$ &$T_n^*$& $Q_{1}$ & $Q_{2}$ & $Q_{3}$ & LM & TB &$T_n$ &$T_n^*$& $Q_{1}$ & $Q_{2}$ & $Q_{3}$ & LM & TB\\

3&2&5.2&5.8&5.2&5.6&5.6&5.2&5.2&3.2&6.6&3.8&3.8&3.8&3.8&4.8&4.0&6.4&4.0&4.0&4.0&5.2&3.8\\
&4&4.6&7.4&3.6&4.4&4.2&4.4&&4.0&7.4&3.2&3.4&3.4&3.6&&3.8&5.4&4.8&5.0&5.0&5.4&\\
&6&5.6&8.6&4.4&5.2&5.0&5.4&&2.8&7.2&3.2&3.6&3.4&3.0&&4.0&5.4&6.0&6.4&6.2&5.2&\\
&8&4.4&8.4&3.6&5.0&4.4&3.0&&3.8&6.2&2.6&3.0&2.8&3.2&&3.8&6.4&5.0&6.8&6.2&4.6&\\
&10&4.2&7.8&3.6&4.4&4.2&4.0&&3.0&6.0&1.4&3.0&2.4&2.4&&3.6&5.6&5.4&7.4&7.2&4.6&\\
15&2&3.8&5.2&4.2&4.8&4.8&5.0&4.8&2.8&4.4&4.2&5.0&5.0&5.4&7.6&3.0&3.8&3.4&4.0&4.0&3.8&5.2\\
&4&4.0&5.4&2.8&5.0&5.0&3.8&&2.6&4.2&2.8&4.6&4.6&3.6&&2.4&4.8&2.2&3.0&3.0&3.2&\\
&6&3.6&6.2&3.2&5.2&5.2&3.8&&2.2&5.2&3.4&5.2&5.0&3.4&&2.0&5.8&1.6&3.2&3.2&2.4&\\
&8&3.6&6.6&2.0&5.2&5.0&1.0&&2.4&6.0&0.8&5.0&4.6&2.0&&2.2&7.2&0.8&2.8&2.8&1.4&\\
&10&3.0&7.0&1.4&5.6&5.2&0.4&&2.2&6.2&1.0&5.0&4.8&1.6&&2.6&6.6&1.0&4.0&3.8&0.8&\\
50&2&2.4&4.0&1.6&2.4&2.4&1.2&8.8&3.0&4.2&1.4&2.4&2.4&1.4&7.8&1.8&4.8&1.6&2.8&2.8&1.2&7.8\\
&4&4.0&4.4&0.6&3.0&2.8&0.0&&2.6&4.6&0.6&2.2&2.2&0.0&&2.2&5.2&0.8&2.6&2.6&0.0&\\
&6&3.6&4.8&0.0&3.8&3.6&&&1.8&5.2&0.2&2.8&2.6&&&2.0&6.4&0.2&2.2&2.2&&\\
&8&3.8&4.4&0.0&3.8&3.6&&&2.0&5.4&0.0&2.2&2.2&&&1.6&7.2&0.0&2.8&2.4&&\\
&10&4.6&4.8&0.0&3.0&3.0&&&1.4&5.4&0.0&2.8&2.2&&&1.4&6.2&0.0&2.0&1.8&&\\
150&2&3.0&4.4&0.0&0.0&0.0&& 0.0&3.0&3.8&0.0&0.2&0.0&&0.0&1.4&3.6&0.0&0.2&0.2&&0.0\\
&4&1.4&4.2&0.0&0.0&0.0&&&2.0&4.2&0.0&0.0&0.0&&&1.4&3.4&0.0&0.0&0.0&&\\
&6&1.8&2.8&0.0&0.0&0.0&&&2.4&3.2&0.0&0.0&0.0&&&1.2&4.2&0.0&0.0&0.0&&\\
&8&2.2&3.8&0.0&0.0&0.0&&&1.8&3.2&0.0&0.2&0.2&&&0.6&4.8&0.0&0.0&0.0&&\\
&10&3.2&4.6&0.0&0.2&0.0&&&1.6&4.2&0.0&0.0&0.0&&&0.4&5.4&0.0&0.0&0.0&&
\end{tabular}
\end{adjustbox}
}
\label{t1} }
\end{table}

\begin{table}[h!]
     {
   \centering
\caption{The empirical sizes ($\%$)  of the tests $T_n$, $T_n^*$,
$Q_{1},  Q_{2}$, $ Q_{3}$, Lagrange multiplier test (LM) and Tiao \& Box'
test (TB) for testing white noise $\varepsilon_t = A z_t$ at the $5\%$ nominal level,
where $z_t$ consists of $p$ independent
autoregressive conditionally heteroscedastic processes.}
      \noindent\makebox[1\textwidth]{%
\begin{adjustbox}{max width=1\textwidth}
\begin{tabular}{cc |ccccccc |ccccccc| ccccccccc}
& &\multicolumn{7}{ c |}{Model 1}& \multicolumn{7}{ c| }{Model 2} & \multicolumn{7}{ c }{Model 3}   \\
 $p$& $K$&  $T_n$ &$T_n^*$& $Q_{1}$ & $Q_{2}$ & $Q_{3}$ & LM & TB & $T_n$ &$T_n^*$& $Q_{1}$ & $Q_{2}$ & $Q_{3}$ & LM & TB &$T_n$ &$T_n^*$& $Q_{1}$ & $Q_{2}$ & $Q_{3}$ & LM & TB\\

3&2&4.0&5.4&3.0&3.0&4.2&4.0&4.4&6.4&8.6&7.0&7.0&9.4&8.4&6.4&3.8&7.2&5.4&5.6&8&6.6&7.2\\
&4&4.4&7.6&4.4&4.6&5.2&5.0&&5.6&8.2&5.2&6.0&7.4&6.2&&5.0&7.8&5.4&6.0&8.2&7.2\\
&6&3.0&6.6&4.8&5.6&6.4&4.6&&5.0&6.6&5.8&6.2&7.2&4.8&&4.8&6.8&4.0&4.4&6.6&4.8\\
&8&3.2&6.4&4.4&5.8&7.4&5.6&&4.6&6.8&5.8&7.0&7.8&6.4&&4.8&6.6&4.2&5.0&5.4&3.4\\
&10&3.6&6.0&5.0&5.8&7.8&5.6&&4.4&6.2&5.4&6.4&7.2&4.2&&4.8&5.8&4.6&5.0&5.4&4.0\\
15&2&4.2&5.6&4.0&3.8&4.6&5.0&7.0&4.8&5.0&3.2&3.2&3.6&3.4&5.6&2.4&4.8&5.0&5.4&6.6&5.2&6.8\\
&4&4.0&5.8&5.0&5.0&5.2&4.0&&3.8&5.0&4.0&4.0&4.0&2.2&&2.6&7.0&2.8&2.8&2.8&3.2\\
&6&4.2&5.0&4.2&4.0&4.2&3.0&&2.8&6.2&5.6&5.0&5.6&2.0&&2.4&6.8&3.0&3.0&3.2&2.6\\
&8&3.8&6.0&4.8&4.8&4.8&1.4&&2.2&5.8&4.8&4.8&4.8&2.0&&2.8&6.2&1.8&2.8&3.8&2.2\\
&10&4.6&4.8&5.6&5.4&5.4&1.2&&3.4&5.4&4.0&3.8&4.4&1.4&&2.4&8.2&0.8&4.2&4.4&0.8\\
50&2&4.4&4.2&2.2&2.2&2.2&0.6&6.2&3.2&4.0&1.4&2.4&2.8&1.0&8.2&2.2&3.2&2.2&2.0&2.0&0.4&7.6\\
&4&3.8&4.6&2.8&2.8&3.0&0.0&&2.2&5.4&2.0&2.0&2.0&0.0&&2.2&4.0&2.0&1.8&1.8&0.0\\
&6&4.6&6.2&1.4&1.4&1.8&&&3.6&5.2&1.8&2.8&1.8&&&1.2&4.8&2.0&2.0&2.0&\\
&8&3.6&7.2&3.0&3.0&3.0&&&2.4&6.0&1.4&1.2&1.6&&&1.6&5.8&0.0&0.0&1.6&&\\
&10&3.6&5.8&3.2&3.2&2.8&&&2.2&5.6&1.8&1.8&1.8&&&1.4&6.6&0.0&0.0&1.6\\
150&2&4.8&3.6&0.0&0.0&0.0&&0.0&1.2&2.8&0.0&0.0&0.0&&0.2&1.6&2.8&0.0&0.0&0.0&&0.0\\
&4&2.8&3.2&0.0&0.0&0.0&&&2.2&3.4&0.0&0.0&0.0&&&1.0&3.2&0.2&0.0&0.2\\
&6&2.0&4.2&0.0&0.0&0.0&&&2.6&3.6&0.0&0.0&0.0&&&1.0&3.2&0.0&0.0&0.0\\
&8&1.6&5.0&0.0&0.0&0.0&&&1.6&4.4&0.0&0.0&0.0&&&0.8&3.4&0.2&0.0&0.2\\
&10&2.6&4.8&0.2&0.2&0.2&&&2.0&5.0&0.4&0.2&0.4&&&1.0&4.6&0.0&0.0&0.0
 \end{tabular}
\end{adjustbox}}
\label{t3} }
 \end{table}

Tables \ref{t1}--\ref{t3} report the empirical sizes
 of tests $T_n$ and $T_n^*$, along with those of
the three portmanteau tests, the Lagrange multiplier test, and the test of \cite{TB81}.
As Tiao \& Box' test does not involve the lag
parameter $K$, we only report its empirical size once for each
$p$ in the tables. Also the Lagrange multiplier test is only applicable
when $pK<n$, as the testing statistic is calculated from
a multivariate regression.

 Tables \ref{t1}--\ref{t3} indicate that  $T_n$ and $T_n^*$
perform about the same as the other five tests when the dimension $p$ is small, such as
$p=3$. The portmanteau, Lagrange multiplier and Tiao \& Box's tests,
however, fail badly to attain
the nominal significance level as the dimension $p$ increases, as the
empirical sizes severely underestimate the nominal level when, for example, $p=50$. In fact the
empirical sizes for the portmanteau tests and Tiao \& Box's  test are almost 0 under
all the settings with $p=150$, while the Lagrange multiplier test, not available
when $p=150$, deviates quickly from
the nominal level when $pK$ is close to $n$.
In contrast, the new test $T_n$ performs much better, though still
underestimates the nominal level when $p$ is relatively large,
particularly for Model 3. Noticeably, $T_n^*$, the procedure
combining the new test with the time series principal component analysis, produces
empirical sizes much closer to the nominal level than all other
tests across almost all the settings with $p=50$ and $150$.

We also observe that both the portmanteau tests $Q_{2}$ and $Q_{3}$ perform
similarly, and outperform $Q_{1}$ when $p$ is large. This
is in line with the fact that the asymptotic approximations for
$Q_{2}$ and $Q_{3}$ are more accurate than that for $Q_{1}$. In addition,
Tables \ref{t1}--\ref{t3}, as well as the results in the Supplementary
Material, indicate that the proposed tests are more robust with respect to the
choice of the prescribed lag parameter $K$.
The test $T_n$, and the portmanteau tests, perform better
under  Models 1 and 2 than under Model 3 when
$p$ is large. As the entries in the loading matrix $A$ in Model 3 can be
both positive and negative, the signals $z_t$ may be
weakened due to possible cancellations. Nevertheless, with the aid
of time series principal component analysis, $T_n^*$ perform reasonably well across
all the settings including Model 3.

In summary, the proposed tests, especially $T_n^*$, attain the nominal
level much more accurate than existing tests when $p$ is large.
For small $p$, all the tests are about equally accurate in attaining
the nominal significance level.

\subsection{Empirical power}

To conduct  the power comparison among the different tests, we consider
two non-white noise models. Put $k_0=\min(\lceil p/5\rceil,12)$.
\begin{description}
\item Model 4: $\varepsilon_t=A
\varepsilon_{t-1}+e_{t}$, where $e_{t}, \, t\ge 1$, are independent,
each $e_{t}$ consists of $p$ independent $t_8$ random variables,
and the coefficient matrix
$A \equiv ( a_{k \ell} )$ is generated as follows:
$a_{k\ell}\sim U(-0.25,0.25)$ independently for $1 \le k,\ell\leq k_0$,
and $a_{k\ell}=0$ otherwise. 
Thus only the first $k_0$ components of $\varepsilon_t$ are not white noise.

\item
Model 5: $\varepsilon_t=Az_t$, where $z_t =(z_{1,t}, \ldots, z_{p,t})^\T$.
For $1\le k \le  k_0 $, $(z_{k,1},\ldots,
z_{k,n})^\T\sim N(0,\Sigma)$, where $\Sigma$ is an $n \times n$
matrix with 1 as the main diagonal elements,  $ 0.5|i-j|^{-0.6}$ as the
$(i,j)$-th element for $1\leq |i-j|\leq 7$, and 0 as all the other elements.
For $k> k_0$, $ z_{k,1},\ldots, z_{k,n}$ are independent and $t_8$ random
variables. The coefficient matrix $A\equiv(a_{k\ell})$ is generated as
follows:
$a_{k\ell}\sim U(-1,1)$ with probability $1/3$ and $a_{k\ell}=0$ with probability $2/3$
independently for $1\le k \ne \ell \le p $, and
$a_{kk}=0.8$ for $1\le k \le p$.

\end{description}

 \begin{figure}
 \begin{center}
        \subfigure{\includegraphics[scale=0.35]{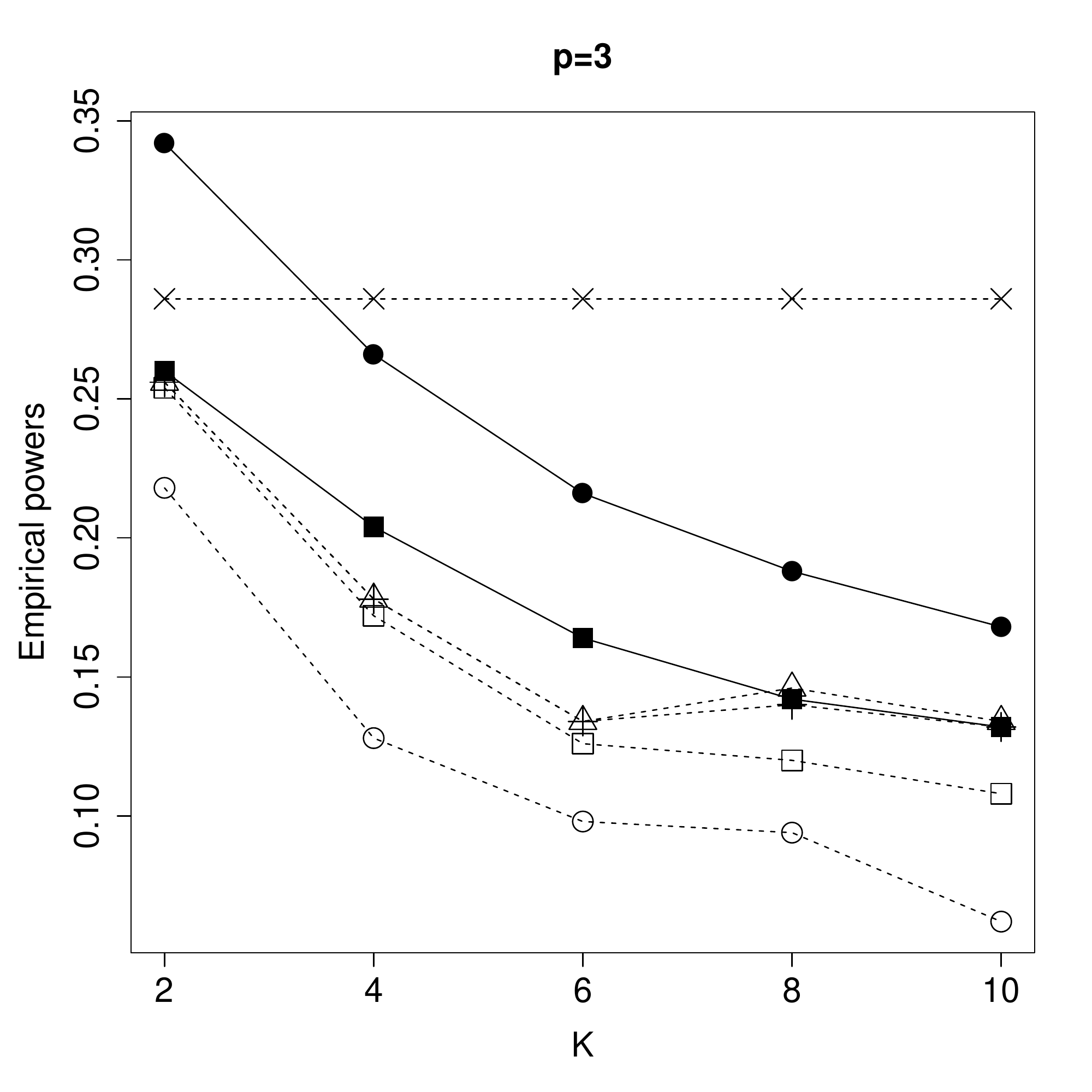}}\subfigure{\includegraphics[scale=0.35]{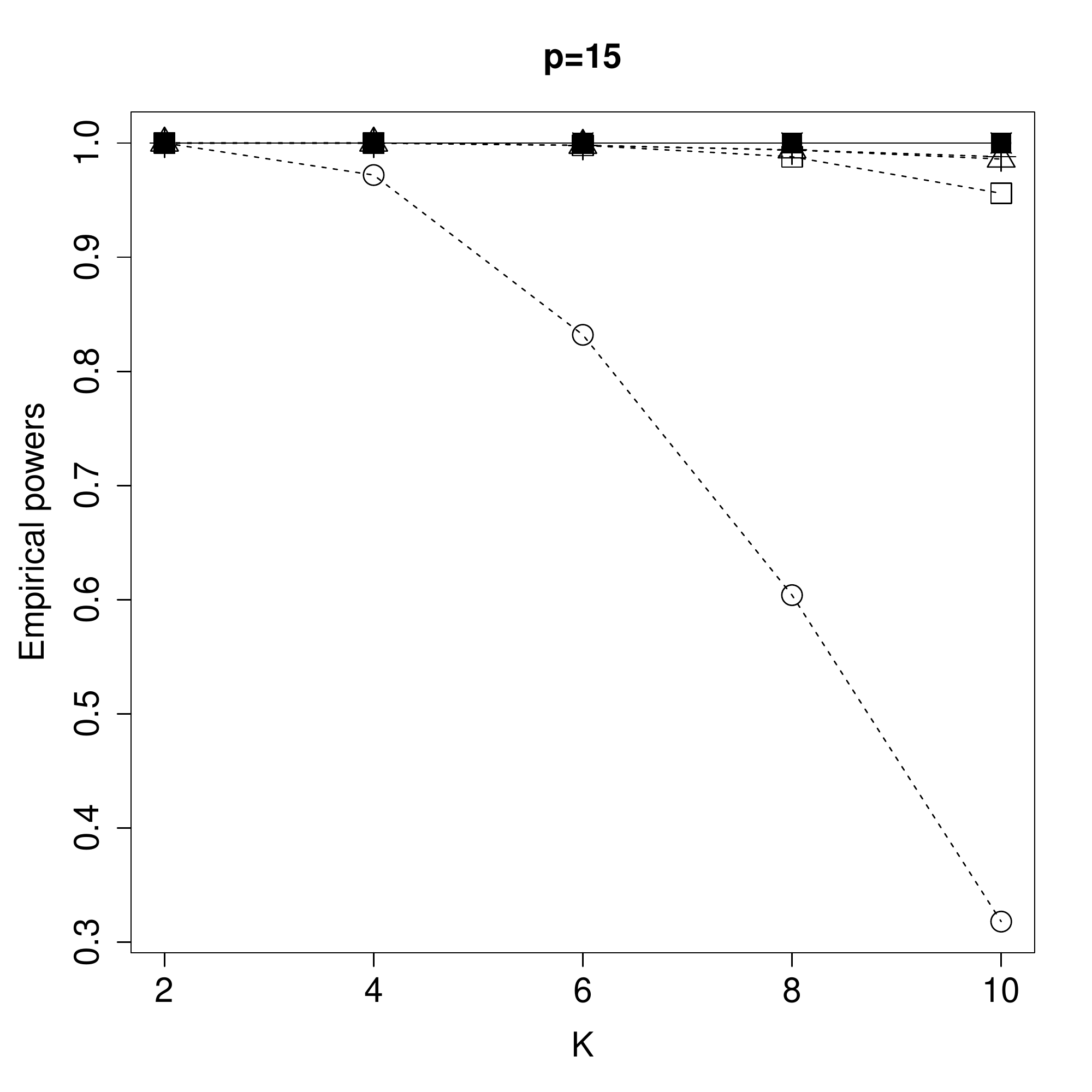}}\\
   \subfigure{\includegraphics[scale=0.35]{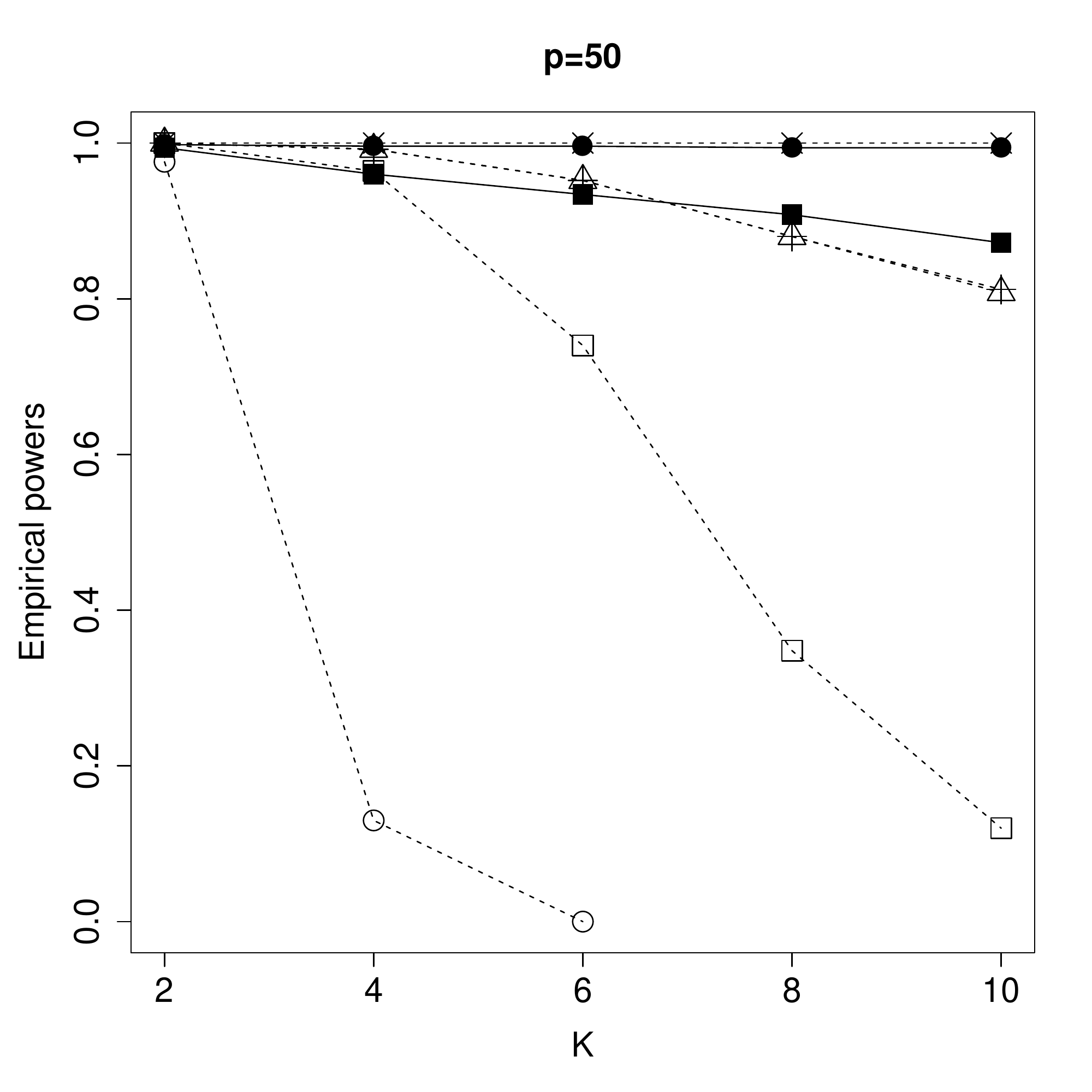}}\subfigure{\includegraphics[scale=0.35]{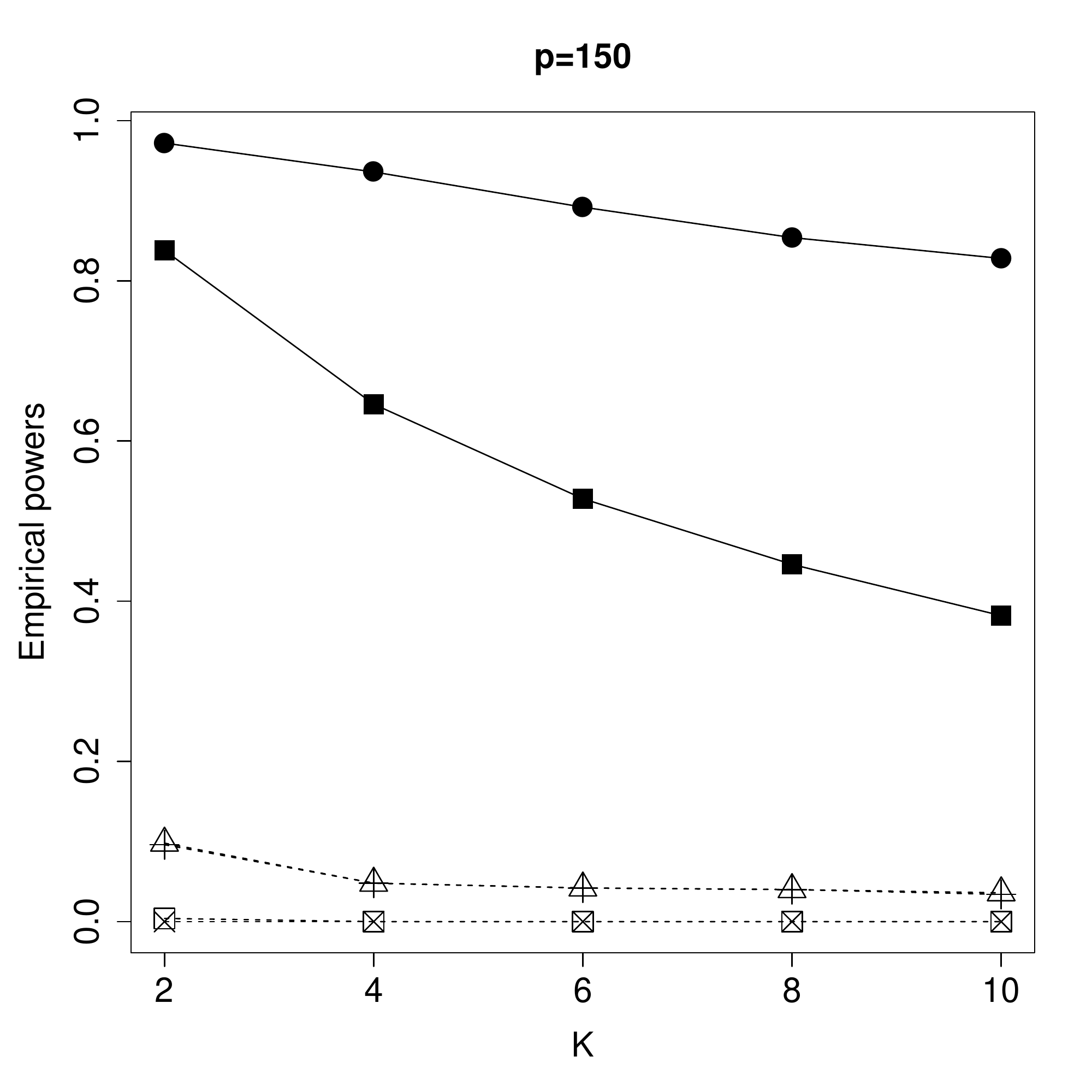}}
   \end{center}
 \caption{Plots of empirical power  against lag $K$ for the new tests
$T_n$ (solid and $\mathsmaller{\blacksquare}$ lines) and $T_n^*$ (solid
and $\bullet$ lines), the portmanteau tests $Q_{1}$ (dashed and
$\vartriangle$ lines), $Q_{2}$ (dashed and $+$ lines) and $Q_{3}$ (dashed
and $\oblong$ lines), the Lagrange multiplier test (dashed and $\circ$
lines), and Tiao and Box' test (dashed and $\times$).  The data are
generated from Model 4 with sample size $n=300$. The nominal level is
$\alpha=5\%$.} 
\label{p1}
    \end{figure}

  \begin{figure}
       \begin{center}
        \subfigure{\includegraphics[scale=0.35]{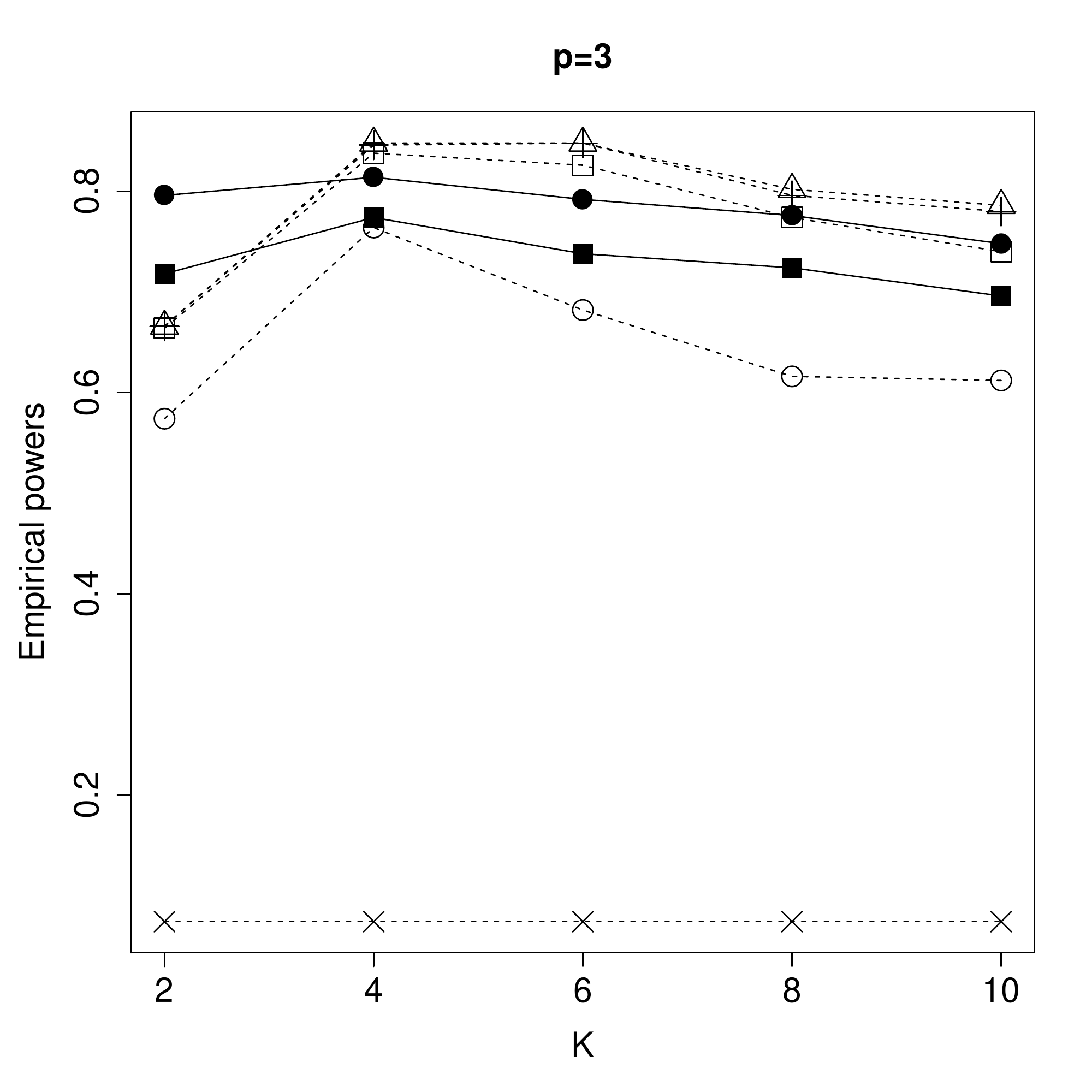}}\subfigure{\includegraphics[scale=0.35]{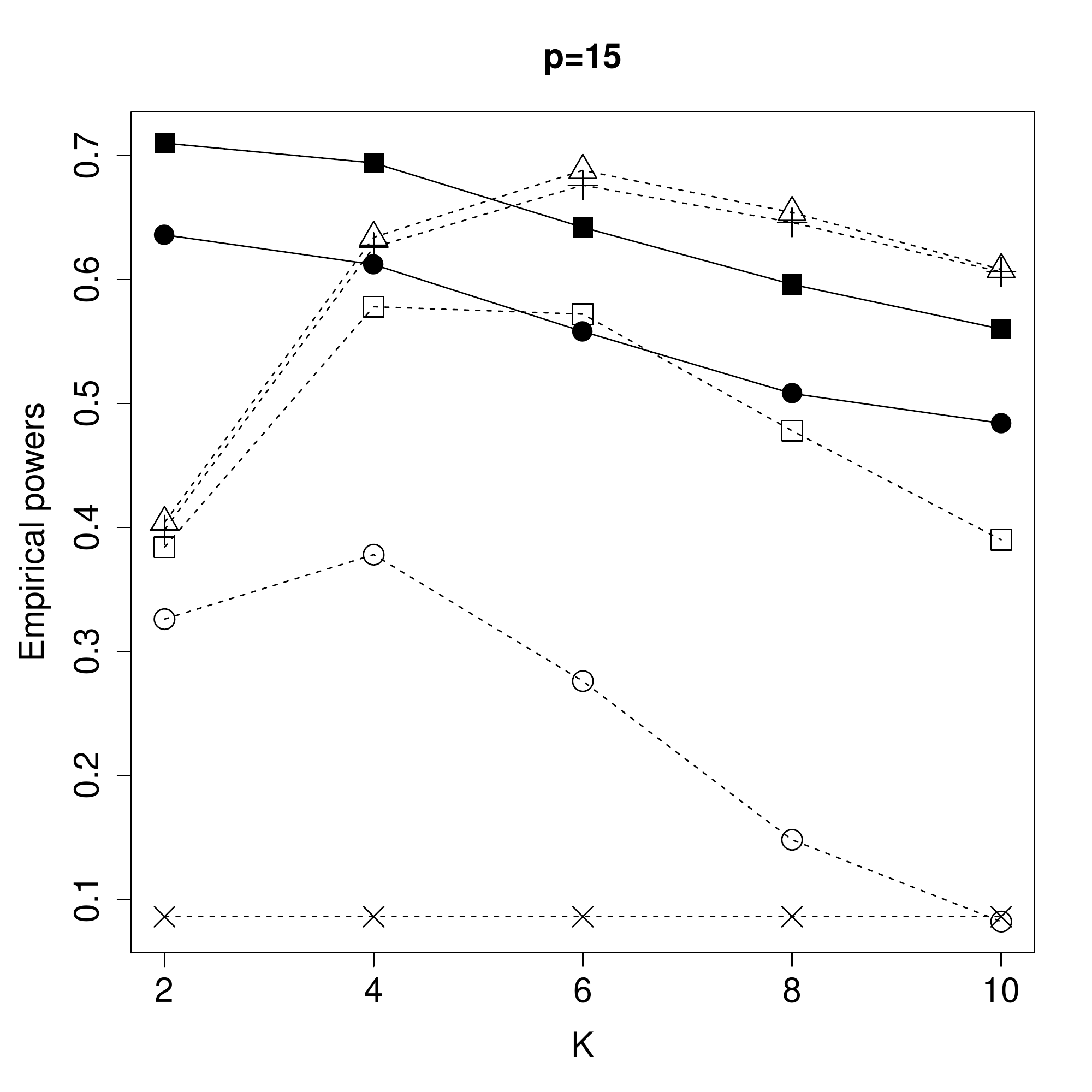}}\\
   \subfigure{\includegraphics[scale=0.35]{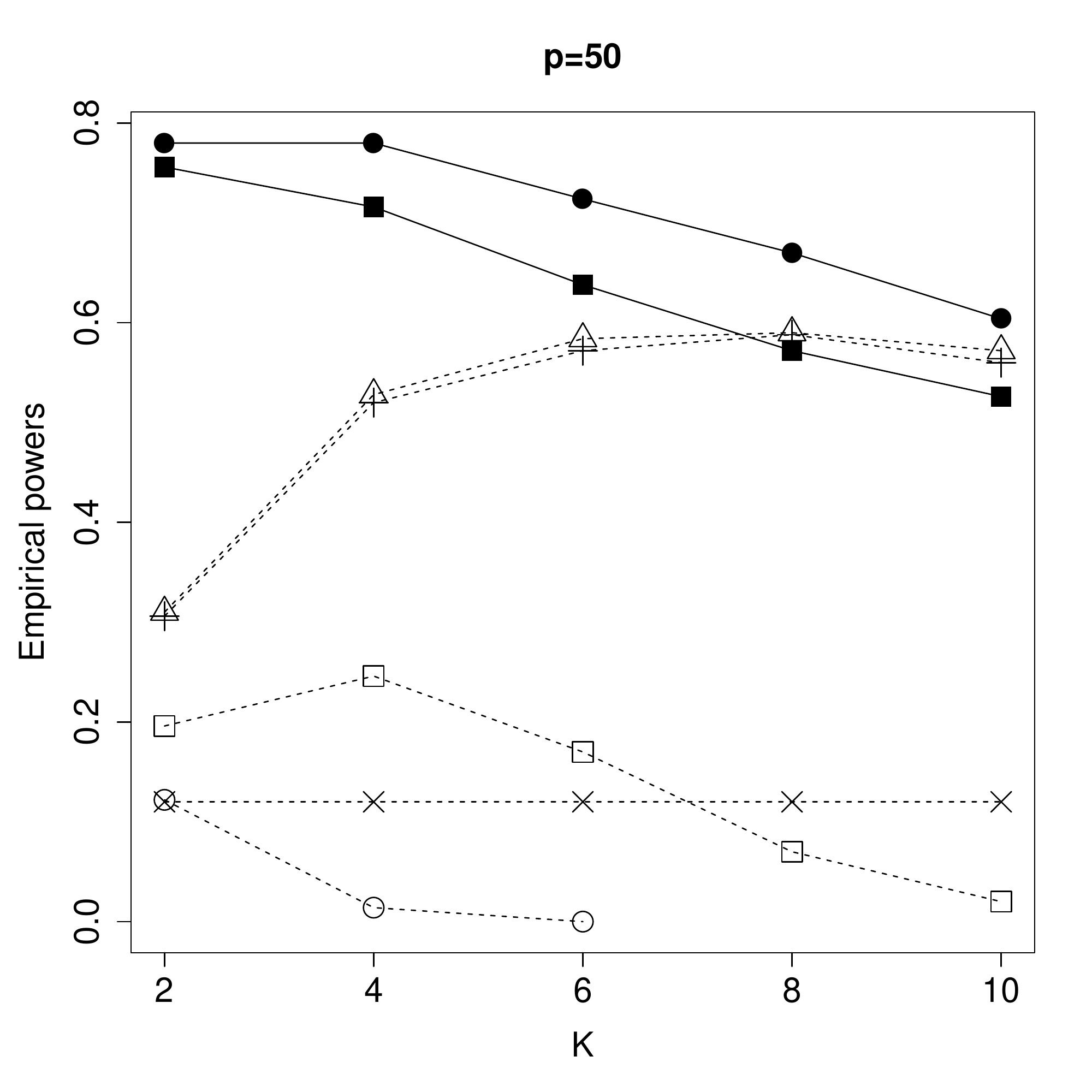}}\subfigure{\includegraphics[scale=0.35]{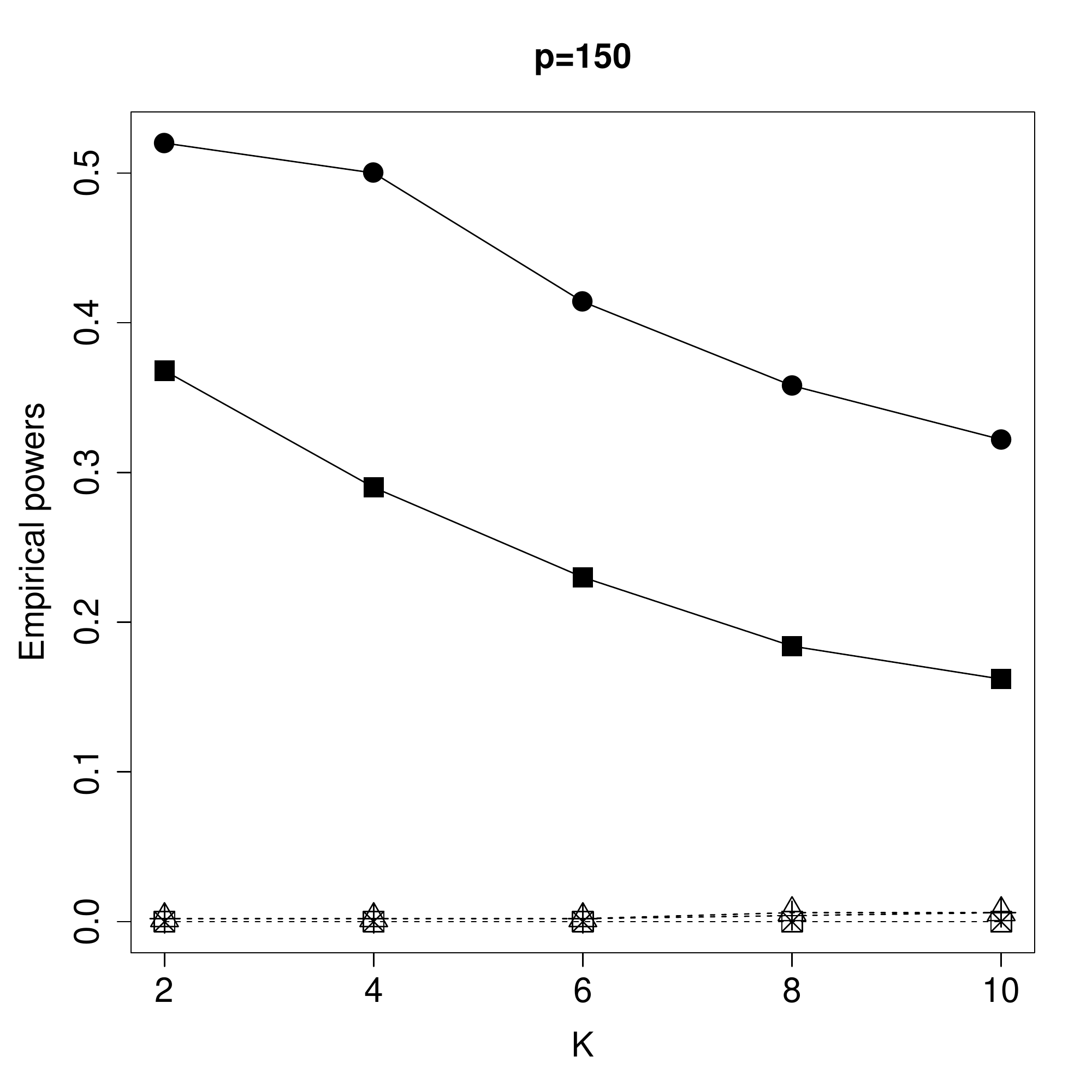}}
   \end{center}
 \caption{Plots of empirical power  against lag $K$ for the new tests
$T_n$ (solid and $\mathsmaller{\blacksquare}$ lines) and $T_n^*$ (solid
and $\bullet$ lines), the portmanteau tests $Q_{1}$ (dashed and
$\vartriangle$ lines), $Q_{2}$ (dashed and $+$ lines) and $Q_{3}$ (dashed
and $\oblong$ lines), the Lagrange multiplier test (dashed and $\circ$
lines), and Tiao and Box' test (dashed and $\times$).  The data are
generated from Model 5 with sample size $n=300$. The nominal level is
$\alpha=5\%$.} 
\label{p4}
\end{figure}

Figs. \ref{p1}--\ref{p4} display the empirical power curves of the seven
tests under consideration against
the lag parameter $K$. 
As Tiao \& Box' test involves no lag parameter $K$, its
power curves are flat.
Also note that the Lagrange multiplier test is only available
for $p=3,15$ and $p=50$ while $K=2,4,6$.
When $p=150$, the proposed tests, especially $T_n^*$, maintain
substantial power while all the other five tests are powerless.
Under Model 4, where the autocorrelation
decays relatively fast, the proposed tests $T_n$ and $T_n^*$  are
substantially more powerful than
the portmanteau tests and the Lagrange multiplier test even when $p$ is
small. In addition, Fig. \ref{p1} and the results in the Supplementary
Materials
also indicate that
the existing tests compromise more in power
than the new tests when the loading matrix $A$ is relatively sparse.
When the autocorrelation is strong, as in Model 5, the portmanteau tests
and the Lagrange multiplier test perform well when $p$ is small, e.g.,
$p=3$; see Fig. \ref{p4}.
Finally, as
expected, $T_n^*$ is more powerful than $T_n$ when $p$ is large, and the
improvement is substantial when, for example, $p=150$. Overall, our
proposed tests $T_n$ and $T_n^*$ are more powerful than the traditional
tests when the dimension $p$ is large or moderately large.
This pattern is also observed in a
more extensive comparison reported in the Supplementary Material.

\section{Applications in model diagnosis}

Let $\{y_t\}$ and $\{u_t\}$ be observable $p\times 1$ and $q\times 1$
time series, respectively.
Let
\begin{equation}\label{eq:general}
y_t=g(u_t;\theta_0)+\varepsilon_t,
\end{equation}
where $g(\cdot;\cdot)$ is a known link function,  and $\theta_0\in\Theta$
is an unknown $s\times 1$
parameter vector.
One of the most frequently used
procedures for model diagnosis is to test
if the error process
$\{\varepsilon_t\}$ is white noise. Since $\{\varepsilon_t\}$ is unknown,
the diagnostic test is instead applied to the residuals
\begin{equation} \label{resid}
\widehat \varepsilon_t \equiv
y_t - g(u_t; \widehat \theta)~~~~(t=1, \ldots, n),
\end{equation}
 where
$\widehat \theta$ is an appropriate estimator for $\theta_0$.

Model \eqref{eq:general} encompasses a large number of frequently used
models, including both linear and 
nonlinear vector autoregressive models
with or without exogenous variables. It also includes linear invertible and identifiable
vector autoregressive and moving average models by allowing $q=\infty$ and $s=\infty$.
Let $g(\cdot;\cdot)=\{g_1(\cdot;\cdot),\ldots,g_p(\cdot;\cdot)\}^\T$,
and $\mathcal{U}$ be the domain of $u_t$.
Let the true value $\theta_0$ of model \eqref{eq:general} be an inner point of $\Theta$.
We assume that the link function
$g(\cdot;\cdot)$ satisfies the following condition.
\begin{condition}\label{as:tes}
Denote by $\Theta_0$ a small neighborhood of $\theta_0$. For some given
metric \linebreak $|\cdot|_*$ defined on $\Theta$, it holds that
$|g_i(u;\theta^*)-g_i(u;\theta^{**})|\leq
M_i(u)|\theta^*-\theta^{**}|_*+R_i(u;\theta^*,\theta^{**})$ for any
$\theta^*, \theta^{**}\in\Theta_0$, $u\in\mathcal{U}$ and $i=1,\ldots,p$,
where $\{M_i(\cdot)\}_{i=1}^p$ and
$\{R_i(\cdot;\cdot,\cdot)\}_{i=1}^p$ are two sets of
non-negative functions that satisfy  $\sup_{1\leq i\leq
p}n^{-1}\sum_{t=1}^nM_i^2(u_t)=O_p(\varphi_{1,n})$ and $\sup_{1\leq i\leq
p}\sup_{\theta^*,\theta^{**}\in\Theta_0}n^{-1}\sum_{t=1}^nR_i^2(u_t;\theta^*,\theta^{**})=O_p(\varphi_{2,n})$
for some $\varphi_{1,n}>0$ (which may diverge) and
$\varphi_{2,n}\rightarrow0$ as $n\rightarrow\infty$.
\end{condition}

In fact, the first part of Condition \ref{as:tes} can be replaced by the
Lipschitz continuity $|g_i(u;\theta^*)-g_i(u;\theta^{**})|\leq
M_i(u)|\theta^*-\theta^{**}|_*^\phi+R_i(u;\theta^*,\theta^{**})$ for some
$\phi\in(0,1]$. Since the proofs for Theorem \ref{tm:4} under these two
types of continuity are identical, we only state the result for $\phi=1$ explicitly.
The remainder term $R_i(\cdot;\cdot,\cdot)$ is employed to accommodate
the models with infinite-dimensional parameter $\theta_0$.
When $\theta_0$ has finite number of components,
we can let $|\cdot|_*$ be
the standard $L_2$-norm. If the link function $g_i(u;\theta)$ is continuously
differentiable with respect to $\theta$, it follows from a Taylor expansion that
$|g_i(u;\theta^*)-g_i(u;\theta^{**})|\leq |\nabla_\theta
g_i(u;\bar{\theta})|_2|\theta^*-\theta^{**}|_2$ for some $\bar{\theta}$
lies between $\theta^{*}$ and
$\theta^{**}$. If there exists an envelop function $M_i(\cdot)$
satisfying $\sup_{\theta\in\Theta}|\nabla_\theta
g_i(u;\bar{\theta})|_2\leq M_i(u)$ for any $u\in\mathcal{U}$, the
first part of Condition \ref{as:tes} holds with $R_i(u;\theta^*,\theta^{**})\equiv0$. When
$\theta_0$ is an infinite dimensional parameter, we can select
$|\cdot|_*$ as the vector $L_1$-norm. Put
$\theta=(\theta_1,\theta_2,\ldots)^\T$. If $\partial
g_i(u;\theta)/\partial\theta_j$ exists for any $j=1,2,\ldots$, it follows from  a Taylor
expansion that
$g_i(u;\theta^*)-g_i(u;\theta^{**})=\sum_{j=1}^\infty(\theta^*_j-\theta^{**}_j)\partial g_i(u;\bar{\theta})/\partial
\theta_j$ for some $\bar{\theta}$ lies between $\theta^*$ and
$\theta^{**}$. For some given diverging $d$, letting $M_i(u)=\sup_{1\leq
j\leq d}\sup_{\theta\in\Theta}|{\partial g_i(u;\theta)}/{\partial
\theta_j}|$ and $R_i(u;\theta^*,\theta^{**})=|\sum_{j=d+1}^\infty
(\theta_j^*-\theta_j^{**}){\partial g_i(u;\bar{\theta})}/{\partial
\theta_j}|$, we have
\begin{align*}
|g_i(u;\theta^*)-g_i(u;\theta^{**})|\leq&~ \sup_{1\leq j\leq d}\bigg|\frac{\partial g_i(u;\bar{\theta})}{\partial \theta_j}\bigg|\sum_{j=1}^d|\theta_j^*-\theta_j^{**}|+\bigg|\sum_{j=d+1}^\infty (\theta_j^*-\theta_j^{**})\frac{\partial g_i(u;\bar{\theta})}{\partial \theta_j}\bigg|\\
\leq&~M_i(u)|\theta^*-\theta^{**}|_1+R_i(u;\theta^*,\theta^{**}).
\end{align*}

\begin{theorem}\label{tm:4}
Let Condition \ref{as:tes} and the conditions of Theorems {\rm\ref{tm:1}}
and {\rm\ref{tm:3}} hold. Let
$|\widehat{\theta}-\theta_0|_*=O_p(\zeta_n)$ for some
$\zeta_n\rightarrow0$. Assume that $\zeta_n^2\varphi_{1,n}\rightarrow0$ as
$n\rightarrow\infty$. Then Theorems {\rm\ref{tm:2}} and
{\rm\ref{tm:3}} still hold if $\{\varepsilon_1, \ldots, \varepsilon_n\}$ is
replaced by $\{\widehat\varepsilon_1, \ldots, \widehat\varepsilon_n\}$ defined in
(\ref{resid}).
\end{theorem}

\section*{Acknowledgement}
We are grateful to the Editor, an Associate Editor and two referees for
their helpful suggestions. Jinyuan Chang was supported in part by the
Fundamental Research Funds for the Central Universities of China,
the National Natural Science Foundation of China, and
the Center of Statistical Research at
Southwestern University of Finance and Economics. Qiwei Yao was supported in part
by the U.K. Engineering and Physical Sciences Research Council.
Wen Zhou was supported in part by the U.S. National Science Foundation. 

\section*{Supplementary material}
Supplementary material available at Biometrika online contains more extensive
comparison by simulation of the seven tests employed in Section 4.


\section*{Appendix}

\setcounter{figure}{0}
\makeatletter
\renewcommand{\thefigure}{S\@arabic\c@figure}
\makeatother

\setcounter{table}{0}
 \makeatletter
\renewcommand{\thetable}{S\@arabic\c@table}
\makeatother

\appendix

\subsection{Technical lemmas}

Let
\[
\widehat{\mu}=[\vecd\{\widehat{\Gamma}(1)\}^\T,\ldots,\vecd\{\widehat{\Gamma}(K)\}^\T]^\T,
\quad \widehat{W}=\textrm{diag}\{\widehat{\Sigma}(0)\}^{-1/2}\otimes\textrm{diag}\{\widehat{\Sigma}(0)\}^{-1/2}.
\]
Then the testing statistic
$T_n={n}^{1/2}|\widehat{\mu}|_\infty. $
It follows from (\ref{eq:autocorr}) that
\[
 \widehat{\mu} \equiv (\widehat{\mu}_1,\ldots,\widehat{\mu}_{p^2K})^\T=(I_K\otimes
\widehat{W})[\vecd\{\widehat{\Sigma}(1)\}^\T,\ldots,\vecd\{\widehat{\Sigma}(K)\}^\T]^\T.
\]
Let
\[
  {\mu}
\equiv({\mu}_1,\ldots,{\mu}_{p^2K})^\T
=(I_K\otimes W)[\vecd\{\widehat{\Sigma}(1)\}^\T,\ldots,
\vecd\{\widehat{\Sigma}(K)\}^\T]^\T,
\]
$$\widehat{Z}={n}^{1/2}\max_{1\leq \ell\leq p^2K}\widehat{\mu}_\ell,\qquad
 Z={n}^{1/2}\max_{1\leq \ell\leq p^2K}\mu_\ell, \qquad V=\max_{1\leq \ell\leq p^2K}G_\ell,
$$
where $G=(G_1,\ldots,G_{p^2K})^\T\sim N(0,\Xi_n)$ with $\Xi_n$ specified
in (\ref{eq:XI}). Throughout the Appendix, $C \in (0, \infty)$ denotes a generic
constant that does not depend on $p$ and $n$, and it may be different at different
places.

\begin{lemma}\label{la:hatw}
Assume that Conditions {\rm\ref{as:var}--\ref{as:betam}} hold. Let
$\gamma$ satisfy $\gamma^{-1}=2r_1^{-1}+r_2^{-1}$, and $\log
p=o\{n^{\gamma/(2-\gamma)}\}$. Then
$ |\widehat{W}-W|_\infty\leq Cn^{-1/2}(\log p)^{1/2}$ with probability at least $1-Cp^{-1}$.
\end{lemma}
\begin{proof} Put
$\diagn\{\widehat{\Sigma}(0)\}=\diagn(\widehat{\sigma}_1^2,\ldots,\widehat{\sigma}_p^2)$ and $\diagn\{\Sigma(0)\}=\diagn(\sigma_1^2,\ldots,\sigma_p^2)$.
By Condition \ref{as:var},
\begin{equation}\label{eq:tr1}
\begin{split}
|\widehat{W}-W|_\infty=&~\max_{1\leq i,j\leq p}|\widehat{\sigma}_i^{-1}\widehat{\sigma}_j^{-1}-\sigma_i^{-1}\sigma_j^{-1}|\leq\bigg(\max_{1\leq i\leq p}|\widehat{\sigma}_i^{-1}-\sigma_i^{-1}|\bigg)^2+C\max_{1\leq i\leq p}|{\widehat{\sigma}_i}^{-1}-{\sigma}_i^{-1}|.
\end{split}
\end{equation}
To bound the term on the right-hand side of \eqref{eq:tr1}, we first
consider the tail probability of $\max_{1\leq i\leq
p}|\widehat{\sigma}_i-\sigma_i|$. Following the same arguments of Lemma 9
in arXiv:1410.2323, it holds that
\begin{align*}
\pr \bigg(\max_{1\leq i\leq p}|\widehat{\sigma}_i^2-\sigma_i^2|>\varepsilon\bigg)
& \leq Cpn\exp(-C\varepsilon^{\gamma}n^{\gamma})+Cpn\exp(-C\varepsilon^{\tilde{\gamma}/2}n^{\tilde{\gamma}})\\
& +Cp\exp(-C\varepsilon^2n)+Cp\exp(-C\varepsilon n)
\end{align*}
for any $\varepsilon>0$ such that $n\varepsilon\rightarrow\infty$,
where $\tilde{\gamma}^{-1}=r_1^{-1}+r_2^{-1}$. Therefore,
if $\log p=o\{n^{\gamma/(2-\gamma)}\}$, with probability at least $1-Cp^{-1}$,
$ \max_{1\leq i\leq p}|\widehat{\sigma}_i^2-\sigma_i^2|\leq Cn^{-1/2}(\log p)^{1/2}. $
Since $\widehat{\sigma}_i^2-\sigma_i^2=(\widehat{\sigma}_i-\sigma_i)^2+
2\sigma_i(\widehat{\sigma}_i-\sigma_i)$, it holds with probability at least $1-Cp^{-1}$
that  $\max_{1\leq i\leq p}|\widehat{\sigma}_i-\sigma_i|\leq Cn^{-1/2}(\log p)^{1/2}$.
Finally, it follows from the identify $\widehat{\sigma}_i^{-1}-\sigma_i^{-1}=-(\widehat{\sigma}_i-\sigma_i)\widehat{\sigma}_i^{-1}\sigma_i^{-1}$
that $ \max_{1\leq i\leq p}|\widehat{\sigma}_i^{-1}-\sigma_i^{-1}|\leq
Cn^{-1/2}(\log p)^{1/2}$ holds with probability at least $1-Cp^{-1}$.
Now the lemma follows from \eqref{eq:tr1} immediately. 
\end{proof}

\begin{lemma}\label{la:2}
Assume that Conditions {\rm\ref{as:var}--\ref{as:betam}} hold. Let
$\gamma^{-1}=2r_1^{-1}+r_2^{-1}$ and $\tilde{\gamma}^{-1}=r_1^{-1}+r_2^{-1}$. Then
\[
\begin{split}
\pr\bigg[\max_{1\leq k\leq K}|\vecd\{\widehat{\Sigma}(k)\}-\vecd\{\Sigma(k)\}|_\infty>s\bigg]\leq&~Cp^2n\exp(-Cs^{\gamma}n^{\gamma})+Cp^2n\exp(-Cs^{\tilde{\gamma}/2}n^{\tilde{\gamma}})\\
&+Cp^2\exp(-Cs^2n)+Cp^2\exp(-Cs n)
\end{split}
\]
for any $s>0$ and $ns\rightarrow\infty$.
\end{lemma}
\begin{proof} Notice that
$|\vecd\{\widehat{\Sigma}(k)\}-\vecd\{\Sigma(k)\}|_\infty=\max_{1\leq
i,j\leq p}|\widehat{\sigma}_{i,j}(k)-\sigma_{i,j}(k)|$. For given
$k=1,\ldots,K$, Lemma 9 in 
arXiv:1410.2323 implies that
\[
\begin{split}
\pr\big[|\vecd\{\widehat{\Sigma}(k)\}-\vecd\{\Sigma(k)\}|_\infty>s\big]
\leq&~Cp^2n\exp(-Cs^{\gamma}n^{\gamma})+Cp^2n\exp(-Cs^{\tilde{\gamma}/2}n^{\tilde{\gamma}})\\
&+Cp^2\exp(-Cs^2n)+Cp^2\exp(-Cs n)
\end{split}
\]
for any $s>0$ and $ns\rightarrow\infty$. Consequently, the lemma
follows directly from the Bonferroni inequality. 
\end{proof}

\begin{lemma} \label{la:z}
Assume that Conditions {\rm\ref{as:var}--\ref{as:betam}} hold. Let
$\gamma^{-1}=2r_1^{-1}+r_2^{-1}$ and $\log p=o\{n^{\gamma/(2-\gamma)}\}$.
Then it holds under null hypothesis $H_0$ that
$
|\widehat{Z}-Z|\leq Cn^{-1/2}\log p
$
 with probability at least $1-Cp^{-1}$.
\end{lemma}
\begin{proof} Note that $ |\widehat{Z}-Z|\leq |\widehat{W}-W|_\infty\max_{1\leq k\leq K}n^{1/2}|\vecd\{\widehat{\Sigma}(k)\}|_\infty. $
By Lemma A\ref{la:2}, we have
$\max_{1\leq k\leq K}|\vecd\{\widehat{\Sigma}(k)\}|_\infty\leq Cn^{-1/2}(\log p)^{1/2}$
with probability at least $1-Cp^{-1}$ under $H_0$.
This, together with Lemma A\ref{la:hatw}, implies the required
assertion.  
\end{proof}

\begin{lemma}\label{la:dn}
Assume that Conditions {\rm\ref{as:var}--\ref{as:cov}} hold.
Let $\log p\leq Cn^{\delta}$ for some $\delta>0$.
Then it holds under $H_0$ that $ \sup_{s\in\mathbb{R}}|\pr(Z\leq s)-\pr(V\leq s)|=o(1). $
\end{lemma}
\begin{proof}
It follow from
(\ref{eq:auto}) that  $
\mu=n^{-1}\sum_{t=1}^{\tilde{n}}u_t+R_n $, where $\tilde{n}=n-K$, each
element of $u_t$ has the form $x_{i,t+k}x_{j,t}/(\sigma_i\sigma_j)$, and
$R_n$ is the remainder term. Let $\tilde{\beta}_k$ $(k\geq1)$ be the
$\beta$-mixing coefficients generated by the process $\{u_t\}$.
Obviously, it holds that $\tilde{\beta}_k\leq \beta_{(k-K)^+}$. Define $
\bar{u}=\tilde{n}^{-1}\sum_{t=1}^{\tilde{n}}u_t\equiv(\bar{u}_1,\ldots,
\bar{u}_{p^2K})^\T$ and $\widetilde{Z}=\tilde{n}^{1/2}\max_{1\leq
\ell\leq p^2K}\bar{u}_\ell. $ In addition, let $
d_n=\sup_{s\in\mathbb{R}}|\pr(Z\leq s)-\pr(V\leq s)|$ and
$\widetilde{d}_n=\sup_{s\in\mathbb{R}}|\pr(\widetilde{Z}\leq s)-\pr(V\leq
s)|. $ We proceed the proof for $d_n=o(1)$ in two steps: (i) to show
$d_n\leq \widetilde{d}_n+o(1)$, and (ii) to prove $\widetilde{d}_n=o(1)$.

To prove (i), note that for any $s\in\mathbb{R}$ and $\varepsilon>0$,
\[
\begin{split}
\pr({Z}\leq s)-\pr(V\leq s)\leq&~ \pr(\widetilde{Z}\leq s+\varepsilon)-\pr(V\leq s+\varepsilon)+\pr(|{Z}-\widetilde{Z}|>\varepsilon)+\pr(s<V\leq s+\varepsilon)\\
\leq&~\widetilde{d}_n+\pr(|{Z}-\widetilde{Z}|>\varepsilon)+\pr(s<V\leq s+\varepsilon).
\end{split}
\]
Similarly, we can obtain the reverse inequality. Therefore,
\begin{equation}\label{eq:dn}
d_n\leq \widetilde{d}_n+\pr(|{Z}-\widetilde{Z}|>\varepsilon)+\sup_{s\in\mathbb{R}}\pr(|V-s|\leq \varepsilon).
\end{equation}
By the anti-concentration inequality of Gaussian random variables, $
\sup_{s\in\mathbb{R}}\pr(|V-s|\leq \varepsilon)\leq
C\varepsilon\{\log(p/\varepsilon)\}^{1/2}. $ It follows from the triangle
inequality and Condition \ref{as:var} that
\[
\begin{split}
|{Z}-\widetilde{Z}|\leq&~ ({n}^{1/2}-\tilde{n}^{1/2})\max_{1\leq\ell\leq p^2K}|\mu_\ell|+\tilde{n}^{1/2}\max_{1\leq \ell\leq p^2K}|\mu_\ell-\bar{u}_\ell|\\
\leq&~\frac{C}{{n}^{1/2}}\max_{1\leq k\leq K}|\vecd\{\widehat{\Sigma}(k)\}|_\infty+\frac{C}{{n}^{1/2}}|\bar{u}|_\infty+{n}^{1/2}|R_n|_\infty.
\end{split}
\]
Following the arguments of Lemma 9 of 
arXiv:1410.2323, we can show that under $H_0$,
\[
\begin{split}
\pr\bigg(\frac{C}{{n}^{1/2}}|\bar{u}|_\infty>\frac{\varepsilon}{3}\bigg)\leq&~ Cp^2n\exp(-C\varepsilon^{\gamma}n^{3\gamma/2})+Cp^2n\exp(-C\varepsilon^{\tilde{\gamma}/2}n^{5\tilde{\gamma}/4})\\
&+Cp^2\exp(-C\varepsilon^2n^2)+Cp^2\exp(-C\varepsilon n^{3/2}) ~,
\end{split}
\]
provided $n^3\varepsilon^2\rightarrow\infty$. It can also shown in the same manner
that under $H_0$, $\pr({n}^{1/2}|R_n|_\infty>{\varepsilon}/{3})$ can be also controlled by the same upper bound specified above. Now  by Lemma A\ref{la:2},  it holds under $H_0$ that
\begin{align*}
\pr(|Z-\widetilde{Z}|>\varepsilon) & \leq  Cp^2n\exp(-C\varepsilon^{\gamma}n^{3\gamma/2})+Cp^2n\exp(-C\varepsilon^{\tilde{\gamma}/2}n^{5\tilde{\gamma}/4})\\
&+Cp^2\exp(-C\varepsilon^2n^2)+Cp^2\exp(-C\varepsilon n^{3/2}).
\end{align*}
Let $\varepsilon=Cn^{-1}(\log p)^{1/2}$. 
Then \eqref{eq:dn} implies that $ d_n\leq \widetilde{d}_n+o(1). $

The proof of (ii) 
is the same as that to show $d_1=o(1)$ in the proof of Theorem 1 of
an unpublished technical report of Chang, Qiu, Yao and Zou (arXiv:1603.06663). Therefore, if $\log p\leq Cn^{\delta}$ for some $\delta>0$, we have $\widetilde{d}_n=o(1)$. This completes the proof of Lemma A\ref{la:dn}. 
\end{proof}

\subsection{Proof of Proposition \ref{tm:1}}

Following the arguments in the proof of Proposition 1 in the
supplementary file of an unpublished technical report of Chang, Zhou
and Zhou (arXiv:1406.1939), it suffices to show $
\sup_{s\in\mathbb{R}}|\pr(\widehat{Z}>s)-\pr(V>s)|=o(1) $, where
$\widehat{Z}$ and $V$ are defined in the first paragraph of Appendix.
Recall $ d_n=\sup_{s\in\mathbb{R}}|\pr(Z\leq s)-\pr(V\leq s)|. $ By
the similar arguments of \eqref{eq:dn}, it can be proved  that $
\sup_{s\in\mathbb{R}}|\pr(\widehat{Z}>s)-\pr(V>s)|
\leq d_n+\pr(|\widehat{Z}-Z|>\varepsilon)+C\varepsilon\{\log(p/\varepsilon)\}^{1/2}. $
Set $\varepsilon=Cn^{-1/2}\log p$, Lemmas A\ref{la:z} and A\ref{la:dn} yield that $ \sup_{s\in\mathbb{R}}|\pr(\widehat{Z}>s)-\pr(V>s)|=o(1). $ This completes the proof of Theorem \ref{tm:1}. 

\subsection{Proof of Theorem \ref{tm:2}}
Based on Lemma 4 
of arXiv:1603.06663 and Proposition \ref{tm:1}, we can
proceed the proof in  the same manner as the proof for Theorem 2 of
arXiv:1603.06663. 

\subsection{Proof of Theorem \ref{tm:3}}

Let $\mathcal{X}_n=\{\varepsilon_1,\ldots,\varepsilon_n\}$. Since
$G\sim N(0,\widehat{\Xi}_n)$ conditionally on $\mathcal{X}_n$, it holds that
$$E(|G|_\infty \mid\mathcal{X}_n)\leq
[1+\{2\log(p^2K)\}^{-1}]\{2\log(p^2K)\}^{1/2}\max_{1\leq \ell \leq
p^2K}\widehat{\Xi}_\ell^{1/2}, $$
where
$\widehat{\Xi}_1,\ldots,\widehat{\Xi}_{p^2K}$ are the elements in the
diagonal of $\widehat{\Xi}_n$. On the other hand, it holds $\pr\{|G|_\infty\geq
E(|G|_\infty\mid\mathcal{X}_n)+u\mid\mathcal{X}_n\}\leq
\exp\{-u^2/(2\max_{1\leq \ell \leq p^2K}\widehat{\Xi}_\ell)\}$ for any
$u>0$. Let ${\Xi}_1,\ldots,{\Xi}_{p^2K}$ be the elements in the main
diagonal of $\Xi_n$. In addition, for any $v>0$, let $
\mathcal {E}_0(v)=\{\max_{1\leq \ell\leq
p^2K}|\widehat{\Xi}_\ell^{1/2}/\Xi_\ell^{1/2}-1|\leq v \}$. Restricted on
$\mathcal{E}_0(v)$, it holds that
$$ \widehat{\cv}_\alpha\leq
(1+v)([1+\{2\log(p^2K)\}^{-1}]\{2\log(p^2K)\}^{1/2}+\{2\log(1/\alpha)\}^{1/2})\max_{1\leq
\ell\leq p^2K}\Xi_\ell^{1/2}.$$
 Let $(i_0,j_0,k_0)=\arg\max_{1\leq k\leq K}\max_{1\leq
 i,j\leq p}|\rho_{i,j}(k)|$. Without loss of generality, we
assume $\rho_{i_0,j_0}(k_0)>0$. Then, restricted on $\mathcal{E}_0(v)$, it holds that
$$
T_n\geq
n^{1/2}\widehat{\rho}_{i_0,j_0}(k_0)
\geq n^{1/2}\widehat{\sigma}_{i_0}^{-1}\widehat{\sigma}_{j_0}^{-1}
\{\widehat{\sigma}_{i_0,j_0}(k_0)-\sigma_{i_0,j_0}(k_0)\}+n^{1/2}\rho_{i_0,j_0}(k_0)(1+v)^{-2}.
$$
Choose $u$ in such a way that
$(1+v)^2[1+\{\log(p^2K)\}^{-1}+u]=1+\varepsilon_n$, for $\varepsilon_n>0$
satisfying that $\varepsilon_n\rightarrow0$ and $\varepsilon_n(\log
p)^{1/2}\rightarrow\infty$. Consequently,
$$
n^{1/2}\rho_{i_0,j_0}(k_0)\geq (1+v)^2[1+\{\log(p^2K)\}^{-1}+u]\lambda(p,\alpha)\max_{1\leq \ell\leq p^2K}\Xi_\ell^{1/2}.
$$
Following the same arguments of Lemma A\ref{la:2}, we can choose suitable $v\rightarrow0$ such that $\pr\{\mathcal{E}_0(v)^c\}\rightarrow0$. Therefore,
\[
\begin{split}
\pr(T_n>\widehat{\cv}_\alpha)\geq&~ \pr\bigg(n^{1/2}\widehat{\rho}_{i_0,j_0}(k_0)>[1+\{\log(p^2K)\}^{-1}]\lambda(p,\alpha)\max_{1\leq \ell\leq p^2K}\Xi_\ell^{1/2}\bigg)\\
\geq &~\pr\bigg[\frac{n^{1/2}\{\widehat{\sigma}_{i_0,j_0}(k_0)-\sigma_{i_0,j_0}(k_0)\}}{\widehat{\sigma}_{i_0}\widehat{\sigma}_{j_0}}>-u\lambda(p,\alpha)\max_{1\leq \ell\leq p^2K}\Xi_\ell^{1/2},~~\mathcal{E}_0(v)~\textrm{holds}\bigg]\\
\geq &~1-\pr\bigg[\frac{n^{1/2}\{\widehat{\sigma}_{i_0,j_0}(k_0)-\sigma_{i_0,j_0}(k_0)\}}{\widehat{\sigma}_{i_0}\widehat{\sigma}_{j_0}}\leq-u\lambda(p,\alpha)\max_{1\leq \ell\leq p^2K}\Xi_\ell^{1/2}\bigg]-\pr\{\mathcal{E}_0(v)^c\}.
\end{split}
\]
Notice that $u\sim \varepsilon_n$. Thus $u\lambda(p,\alpha)\max_{1\leq
\ell\leq p^2K}\Xi_\ell^{1/2}\rightarrow\infty$, which implies that $ \pr(T_n>\widehat{\cv}_\alpha)\rightarrow1. $ This completes the proof of Theorem \ref{tm:3}. 

\subsection{Proof of Theorem \ref{tm:4}}
Let $\widehat{W}^*$, $\widehat{\Sigma}^*(0)$, $\widehat{J}_n^*$ and
$\widehat{\Xi}_n^*$ be, respectively, the analogues of $\widehat{W}$,
$\widehat{\Sigma}(0)$, $\widehat{J}_n$ and $\widehat{\Xi}_n$ with
$\varepsilon_t$ replaced by $\widehat{\varepsilon}_t$. By Lemma 3.1 of
\cite{CCK_2013}, 
we only need to show
$|\widehat{\Xi}_n^*-\widehat{\Xi}_n|_\infty=o_p(1)$. Recall
$\widehat{\Xi}_n=(I_K\otimes\widehat{W})\widehat{J}_n(I_K\otimes\widehat{W})$
and $\widehat{\Xi}_n^*=(I_K\otimes\widehat{W}^*)\widehat{J}_n^*(I_K\otimes\widehat{W}^*)$,
it suffices to prove $|\widehat{W}^*-\widehat{W}|_\infty=o_p(1)$ and
$|\widehat{J}_n^*-\widehat{J}_n|_\infty=o_p(1)$.
Since the proofs for those two assertions are similar,
we only present the proof for $|\widehat{W}^*-\widehat{W}|_\infty=o_p(1)$ below. As
$\widehat{W}=[\diagn\{\widehat{\Sigma}(0)\}]^{-1/2}\otimes
[\diagn\{\widehat{\Sigma}(0)\}]^{-1/2}$ and
$\widehat{W}^*=[\diagn\{\widehat{\Sigma}^*(0)\}]^{-1/2}\otimes [\diagn\{\widehat{\Sigma}^*(0)\}]^{-1/2}$, it suffices to show $|\widehat{\Sigma}^*(0)-\widehat{\Sigma}(0)|_\infty=o_p(1)$.
Put $\widehat{\varepsilon}_t=(\widehat{\varepsilon}_{1,t},\ldots,
\widehat{\varepsilon}_{p,t})^\T$ and ${\varepsilon}_t=({\varepsilon}_{1,t},
\ldots,{\varepsilon}_{p,t})^\T$. For any $i,j$, the $(i,j)$-th element of
$\widehat{\Sigma}^*(0)-\widehat{\Sigma}(0)$ is given by
$\Delta_{i,j}=n^{-1}\sum_{t=1}^n(\widehat{\varepsilon}_{i,t}
\widehat{\varepsilon}_{j,t}-\varepsilon_{i,t}\varepsilon_{j,t})$.
Notice that
$\widehat{\varepsilon}_{i,t}=y_{i,t}-g_i(u_t;\widehat{\theta})$ and
$\varepsilon_{i,t}=y_{i,t}-g_i(u_t;\theta_0)$. It holds that
\begin{align*}
\Delta_{i,j}=&~ \frac{1}{n}\sum_{t=1}^n\{g_i(u_t;\widehat{\theta})-g_i(u_t;\theta_0)\}
\{g_j(u_t;\widehat{\theta})-g_j(u_t;\theta_0)\}\\
&  -\frac{1}{n}\sum_{t=1}^n\{g_i(u_t;\widehat{\theta})-g_i(u_t;\theta_0)\}\varepsilon_{j,t}-\frac{1}{n}\sum_{t=1}^n\varepsilon_{i,t}\{g_j(u_t;\widehat{\theta})-g_j(u_t;\theta_0)\}.
\end{align*}
It follows from Cauchy--Schwarz inequality that
\begin{equation}\label{eq:bound}
\begin{split}
\Delta_{i,j}^2\leq&~ 3\bigg[\frac{1}{n}\sum_{t=1}^n\{g_i(u_t;\widehat{\theta})-g_i(u_t;\theta_0)\}^2\bigg]\bigg[\frac{1}{n}\sum_{t=1}^n\{g_j(u_t;\widehat{\theta})-g_j(u_t;\theta_0)\}^2\bigg]\\
&+3\bigg[\frac{1}{n}\sum_{t=1}^n\{g_i(u_t;\widehat{\theta})-g_i(u_t;\theta_0)\}^2\bigg]\bigg(\frac{1}{n}\sum_{t=1}^n\varepsilon_{j,t}^2\bigg)\\
&+3\bigg[\frac{1}{n}\sum_{t=1}^n\{g_j(u_t;\widehat{\theta})-g_j(u_t;\theta_0)\}^2\bigg]\bigg(\frac{1}{n}\sum_{t=1}^n\varepsilon_{i,t}^2\bigg).
\end{split}
\end{equation}
By Condition \ref{as:tes}, it holds uniformly for any $i=1,\ldots,p$ that
\begin{align*}
{1 \over n}\sum_{t=1}^n\{g_i(u_t;\widehat{\theta})-g_i(u_t;\theta_0)\}^2 &
\leq |\widehat{\theta}-\theta_0|_*^2\bigg\{\frac{2}{n}\sum_{t=1}^nM_i^2(u_t)\bigg\}+\frac{2}{n}\sum_{t=1}^nR_i^2(u_t;\widehat{\theta},\theta_0)\\
& =O_p(\zeta_n^2\varphi_{1,n}+\varphi_{2,n}).
\end{align*}
On the other hand, Lemma A\ref{la:2} implies that $\sup_{1\leq i\leq
p}n^{-1}\sum_{t=1}^n\varepsilon_{i,t}^2=O_p(1)$. This together with \eqref{eq:bound} imply that
 $\Delta_{ij}^2=O_p(\zeta_n^2\varphi_{1,n}+\varphi_{2,n})$  uniformly for any $i,j=1,\ldots,p$. Thus $|\widehat{\Sigma}^*(0)-\widehat{\Sigma}(0)|_\infty=O_p(\zeta_n\varphi_{1,n}^{1/2}+\varphi_{2,n}^{1/2})=o_p(1)$. This completes the proof of Theorem \ref{tm:4}.

\bibliographystyle{biometrika}

\end{document}